\documentclass{article}
\usepackage[T1]{fontenc}
\usepackage[utf8]{inputenc}
\usepackage{amsmath}
\usepackage{amssymb}
\usepackage{amsthm}
\usepackage{latexsym}
\usepackage{bussproofs}
\usepackage{cmll}

\usepackage{fullpage}
\usepackage{array}
\usepackage{alltt}

\usepackage{tikz}
\usetikzlibrary{matrix}
\usetikzlibrary{cd}
\usepackage{pgfplots}

\usepackage{hyperref}

\newcommand\CMLLPAR{
\usepackage{cmll}
\newcommand\IPar{\mathord{\parr}}
}

\CMLLPAR

\title{Upper approximating probabilities of convergence in probabilistic coherence spaces}

\author{Thomas Ehrhard\\
  Université de Paris, IRIF, CNRS, F-75013 Paris, France\\
  \url{ehrhard@irif.fr}}

\usepackage{stmaryrd}

\newtheorem*{lemma*}{Lemma}
\newtheorem*{proposition*}{Proposition}

\newcommand{\proofitem}[1]{\paragraph*{\mdseries\textit{#1}}}

\newcommand{\Beginproof}{\proofitem{Proof.}}
\newcommand{\Endproof}{
  \ifmmode % if math mode, assume display: omit penalty etc.
  \else \leavevmode\unskip\penalty9999 \hbox{}\nobreak\hfill
  \fi
  \quad\hbox{$\Box$}
  \par\medskip}

\newcommand\Eqref[1]{(\ref{#1})}

\newcommand\Eg{\textsl{e.g.}}

\renewcommand{\phi}{\varphi}
\renewcommand\epsilon{\varepsilon}

\newcommand{\Implies}{\Rightarrow}

\newcommand\Equiv{\Leftrightarrow}
\newcommand{\St}{\mid}

\newcommand{\arrow}{\rightarrow}
\newcommand{\Llbot}{{\mathord{\perp}}}
\newcommand{\Top}{\top}

\newcommand\cE{\mathcal{E}}

\newcommand\cK{\mathcal{K}}

\newcommand\cP{\mathcal{P}}

\newcommand\cS{\mathcal{S}}

\newcommand\cU{\mathcal{U}}

\newcommand\Fini{{\mathsf{fin}}}

\newcommand\Union{\bigcup}

\newcommand{\Linarrow}{\multimap}

\newcommand\Myleft{}
\newcommand\Myright{}

\newcommand\Web[1]{\Myleft|{#1}\Myright|}

\newcommand\Supp[1]{\operatorname{\mathsf{supp}}({#1})}

\newcommand\Mset[1]{[{#1}]}

\newcommand\ITens{\otimes}
\newcommand\Tens[2]{{#1}\ITens{#2}}
\newcommand\Tensp[2]{\left({#1}\ITens{#2}\right)}
\newcommand\IWith{\mathrel{\&}}
\newcommand\With[2]{{#1}\IWith{#2}}
\newcommand\IPlus{\oplus}
\newcommand\Plus[2]{{#1}\IPlus{#2}}
\newcommand\Orth[2][]{#2^{\Llbot_{#1}}}
\newcommand\Orthp[2][]{(#2)^{\Llbot_{#1}}}

\newcommand\Bwith{\mathop{\&}}
\newcommand\Bplus{\mathop\oplus}
\newcommand\Bunion{\mathop\cup}

\newcommand\Biorth[1]{#1^{\Llbot\Llbot}}

\newcommand\Triorth[1]{{#1}^{\Llbot\Llbot\Llbot}}

\newcommand\One{1}

\newcommand\Locun[1]{1^J}

\newcommand\Isom\simeq

\newcommand\Funinv[1]{{#1}^{-1}}

\newcommand\Limpl[2]{{#1}\Linarrow{#2}}
\newcommand\Limplp[2]{(\Limpl{#1}{#2})}

\newcommand\Nat{{\mathbb{N}}}

\newcommand\Biind[2]{\genfrac{}{}{0pt}{1}{#1}{#2}}

\newcommand\Snat{\Flat\Nat}

\newcommand\App[2]{({#1}){#2}}

\newcommand\Abst[3]{\lambda#1^{#2}\,{#3}}

\newcommand\List[3]{#1_{#2},\dots,#1_{#3}}

\newcommand\Kronecker[2]{\delta_{{#1},{#2}}}

\newcommand\Subst[3]{{#1}\left[{#2}/{#3}\right]}

\newcommand\Substbis[2]{{#1}\left[{#2}\right]}

\newcommand\Real{\mathbb{R}}
\newcommand\Realp{\mathbb{R}_{\geq 0}}
\newcommand\Realpto[1]{(\Realp)^{#1}}

\newcommand\Realpc{\overline{\Realp}}
\newcommand\Realpcto[1]{\Realpc^{#1}}

\newcommand\Rational{\mathbb Q}

\newcommand\Mfin[1]{\mathcal M_\Fini({#1})}

\newcommand\Ev{\operatorname{\mathsf{Ev}}}
\newcommand\Evlin{\operatorname{\mathsf{ev}}}

\newcommand\Norm[1]{\|{#1}\|}
\newcommand\Normsp[2]{\|{#1}\|_{#2}}

\newcommand\Rel[1]{\mathrel{#1}}

\newcommand\Redst[1]{\mathop{\mathsf{Red}}}

\newcommand\Tuple[1]{\langle{#1}\rangle}

\newcommand\Msetofsubst[1]{\bar F}

\newcommand\Pcoh[1]{\mathsf P{#1}}
\newcommand\Pcohp[1]{\Pcoh{\left(#1\right)}}
\newcommand\Pcohc[1]{\overline{\mathsf P}{#1}}

\newcommand\Base[1]{e_{#1}}
\newcommand\Matapp[2]{{#1}\Compl{#2}}
\newcommand\Matappa[2]{{#1}\cdot{#2}}

\newcommand\PCOH{\mathbf{Pcoh}}

\newcommand\Leftu{\lambda}
\newcommand\Rightu{\rho}
\newcommand\Assoc{\alpha}
\newcommand\Sym{\gamma}

\newcommand\Retri\zeta
\newcommand\Retrp\rho

\newcommand\Impl[2]{{#1}\Rightarrow{#2}}

\newcommand\Tsem[1]{\llbracket{#1}\rrbracket}

\newcommand\Psem[2]{\llbracket{#1}\rrbracket_{#2}}
\newcommand\Psemst[2]{\llbracket{#1}\rrbracket_{#2}}
\newcommand\Psemstate[2]{\llbracket{#1}\rrbracket_{#2}}

\newcommand\Tnat\iota
\newcommand\Fix[1]{\operatorname{\mathsf{fix}}(#1)}
\newcommand\If[3]{\operatorname{\mathsf{if}}(#1,#2,#3)}
\newcommand\Pred[1]{\operatorname{\mathsf{pred}}(#1)}
\newcommand\Succ[1]{\operatorname{\mathsf{succ}}(#1)}
\newcommand\Num[1]{\underline{#1}}
\newcommand\Loop\Omega

\newcommand\Dice[1]{\operatorname{\mathsf{coin}}(#1)}
\newcommand\Tseq[3]{{#1}\vdash{#2}:{#3}}
\newcommand\Tseqst[3]{{#1,#2:#3}\vdash}

\newcommand\Tseqstate[2]{#1\vdash #2}

\newcommand\Timpl\Impl
\newcommand\Simpl\Impl

\newcommand\PCF{\mathsf{PCF}}

\newcommand\Fixpcoh[1]{\operatorname{\mathsf{Fix}}_{#1}}

\newcommand\Weak[1]{\operatorname{\mathsf{w}}_{#1}}
\newcommand\Contr[1]{\operatorname{\mathsf{contr}}_{#1}}

\newcommand\Der[1]{\operatorname{\mathsf{der}}_{#1}}

\newcommand\Digg[1]{\operatorname{\mathsf{dig}}_{#1}}

\newcommand\Fun[1]{\widehat{#1}}

\newcommand\Id{\operatorname{\mathsf{Id}}}

\newcommand\Proj[1]{\pi_{#1}}

\newcommand\Excl[1]{\oc{#1}}
\newcommand\Exclp[1]{\oc({#1})}

\newcommand\Prom[1]{#1^!}
\newcommand\Promm[1]{{#1}^{!!}}

\newcommand\Relincl\eta
\newcommand\Relrestr\rho

\newcommand\Seelyz{\mathsf m^0}
\newcommand\Seelyt{\operatorname{\mathsf m}^2}

\newcommand\Compl{\,}
\newcommand\Curlin{\operatorname{\mathsf{cur}}}

\newcommand\Kl[1]{{#1}_\oc}

\newcommand\Eval[2]{\langle#1,#2\rangle}

\newcommand\Let[3]{\mathsf{let}(#1,#2,#3)}

\newcommand\Ssuc{\overline{\mathsf{suc}}}
\newcommand\Spred{\overline{\mathsf{pred}}}

\newcommand\Vect[1]{\overrightarrow{#1}}

\newcommand\Bnfeq{\mathrel{\mathord:\mathord=}}
\newcommand\Bnfor{\,\,\mathord|\,\,}

\newcommand\Stcons{\cdot}
\newcommand\Stsucc[1]{\mathsf{succ}\Stcons #1}
\newcommand\Stpred[1]{\mathsf{pred}\Stcons #1}
\newcommand\Stif[3]{\mathsf{if}(#1,#2)\Stcons #3}

\newcommand\Stlet[3]{\mathsf{let}(#1,#2)\Stcons #3}
\newcommand\State[2]{\langle#1,#2\rangle}

\newcommand\Stempty{\epsilon}

\newcommand\Len[1]{\mathsf{len}(#1)}

\newcommand\Extleq{\sqsubseteq}
\newcommand\Extc[1]{\underline{#1}}
\newcommand\Extd[1]{\cE_{#1}}
\newcommand\Exto[1]{\Extleq_{#1}}

\newcommand\EPCOH{\PCOH^{\mathsf e}}

\newcommand\Eset[1]{\{#1\}}

\newcommand\Klapp[2]{\Fun{#1}(#2)}

\newcommand\Nerr{\top}
\newcommand\Invi[2]{\mathsf{inv}(#1,#2)}

\newcommand\Starg[2]{\mathsf{arg}(#1).#2}

\newcommand\Errdiv{\Omega}
\newcommand\Errconv{\mho}

\newcommand\Intercc[2]{[#1,#2]}

\newcommand\Flat[1]{#1_{\mathord\perp}}
\newcommand\Flate[1]{\Flat{#1}^\Nerr}

\newcommand\Pcasee[1]{\widetilde{\mathsf{case}}^{#1}}

\newcommand\Funpcasee[1]{\widehat{\mathsf{case}}^{#1}}

\newcommand\Slete[1]{\widetilde{\mathsf{let}}\left(#1\right)}

\newcommand\Redwhd{\beta_{\mathsf{wh}}}
\newcommand\Redwhp[1]{\beta_{\mathsf{wh}}^{#1}}

\newcommand\Figbreak{\\[3pt]}

\newcommand\Cpathl[1]{\Len{#1}}
\newcommand\Cpaths[1]{\operatorname{\mathsf s}(#1)}
\newcommand\Cpatht[1]{\operatorname{\mathsf t}(#1)}
\newcommand\Cpathp[1]{\operatorname{\mathsf{pr}}(#1)}
\newcommand\Cpathset[2]{\operatorname{\mathsf{cp}}\left(#1,#2\right)}

\newcommand\Termsty[2]{\langle#1\vdash #2\rangle}
\newcommand\Termstyf[2]{\langle#1\vdash #2\rangle_\Fini}

\newcommand\Whnormal{$\mathsf{wh}$-normal}

\newcommand\Redwhpr[2]{\operatorname{\mathbb{P}}\left(#1\downarrow#2\right)}
\newcommand\Redwhprc[1]{\operatorname{\mathbb{P}}_\Errconv\left(#1\right)}

\newcommand\Termso{\Extleq}

\newcommand\Bfune[1]{\widetilde{#1}}

\newcommand\Fsucc{\operatorname{\mathsf{s}}}
\newcommand\Fpred{\operatorname{\mathsf{p}}}

\newcommand\Kreval{\operatorname{\cK}}

\newcommand\Redintst[1]{\|#1\|}
\newcommand\Redint[1]{|#1|}

\newcommand\Fixfree{fix-free}

\newcommand\Bb[1]{\mathbb{#1}}

\newcommand\Polynoms[3]{#1\left[#2,#3\right]}
\newcommand\Polynomts[3]{#1\left[#2,#3\right]_{\mathsf{wf}}}

\newcommand\Ptone{1}
\newcommand\Ptzero{0}
\newcommand\Ptvar[2]{#1\triangleleft #2}
\newcommand\Ptcl[4]{[#1\cdot#2,#3\cdot#4]}

\newcommand\Polyofpt{\operatorname{\mathsf{pol}}}

\newcommand\Idl[2]{\Id^{\mathord\sqsubseteq}_{#1,#2}}
\newcommand\Idu[2]{\Id^{\mathord\sqsupseteq}_{#1,#2}}

\newcommand\Meone[2]{{#1}\cdot{#2}}

\newcommand\Natt{\Nat^{\Nerr}}

\newcommand\Scatch[1]{\overline\epsilon_{#1}}

\newcommand\Pole[1]{\mathord\Bot_{#1}}

\newcommand\Andprop{\wedge}

\newtheorem{theorem}{Theorem}

\newtheorem{proposition}[theorem]{Proposition}
\newtheorem{lemma}[theorem]{Lemma}
\newtheorem{example}[theorem]{Example}

\newenvironment{remark}%
{\smallbreak\noindent{\bf Remark.}\nobreak}%
{\normalsize\smallbreak}

\newcommand\Itmath{\par\noindent$\blacktriangleright$\ }

\begin{document}
\maketitle

\begin{abstract}
    We develop a theory of probabilistic coherence spaces equipped with
  an additional extensional structure and apply it to approximating
  probability of convergence of ground type programs of probabilistic
  PCF whose free variables are of ground types.  To this end we define
  an adapted version of Krivine Machine which computes polynomial
  approximations of the semantics of these programs in the model.  These
  polynomials provide approximations from below and from above of
  probabilities of convergence; this is made possible by extending the
  language with an error symbol which is extensionally maximal in the
  model.

\end{abstract}

\section*{Introduction}

Various settings are now available for the denotational
interpretations of probabilistic programming languages.
\begin{itemize}
\item \emph{Game} based models, first proposed
  in~\cite{DanosHarmer00} and further developed by various authors
  (see~\cite{CastellanClairambaultPaquetWinskel18} for an example of
  this approach). From their deterministic ancestors they typically
  inherit good definability features.
\item Models based on Scott continuous functions on domains endowed
  with additional probability related structures. Among these models
  we can mention %Plotkin and Keimel
  \emph{Kegelspitzen}~\cite{KeimelPlotkin17} (domains equipped with an
  algebraic convex structure) and \emph{$\omega$-quasi Borel
    spaces}~\cite{VakarKammarStaton19} (domains equipped with a
  generalized notion of measurability). In contrast with the former,
  the latter uses an adapted probabilistic powerdomain construction.
\item Models based on (a generalization of) Berry stable
  functions. The first category of this kind was that of
  \emph{probabilistic coherence spaces} (PCSs) and power series with
  non-negative coefficients (the Kleisli category of the model of
  Linear Logic developed in~\cite{DanosEhrhard08}) for which could be
  proved adequacy and full abstraction with respect to a probabilistic
  version of $\PCF$~\cite{EhrhardPaganiTasson18}. This
  setting was extended to ``continuous data types'' (such as $\Real$)
  by substituting PCSs with \emph{positive cones} and power series
  with functions featuring an hereditary monotonicity
  called~\emph{stability}~\cite{EhrhardPaganiTasson18b}. 
\end{itemize}
Just as games, probabilistic coherence spaces interpret types using
simply defined combinatorial devices called \emph{webs} which are
countable sets very similar to game arenas, and even simpler since the
order of basic actions is not taken into account. A closed program of
type $\sigma$ is interpreted as a map from the web associated with
$\sigma$ to the non-negative real half-line $\Realp$. Mainly because
of the absence of explicit sequencing information in web elements,
these functions are not always probability sub-distributions. For
instance, a closed term $M$ of type $\Timpl\Tnat\One$ (integer to unit
type) is interpreted as a mapping from $\Mfin\Nat$ (finite multisets
of integers) to $\Realp$, rather seen as an indexed family
$S=(\alpha_\mu)_{\mu\in\Mfin\Nat}$. Given $u\in\Realp^\Nat$ and
$\mu\in\Mfin\Nat$ one sets\footnote{In other words, we consider $\mu$
  as a multiexponent for ``variables'' indexed by $\Nat$ and $u$ is
  considered as a valuation for these variable.}
$u^\mu=\prod_{n\in\Nat}u_n^{\mu(n)}$ (where $\mu(n)$ is the
multiplicity of $n$ in $\mu$) and then we can see $S$ as the function
$\Fun S:\Pcoh\Nat\to\Intercc 01$ defined by
$\Fun S(u)=\sum_{\mu\in\Mfin\Nat}\alpha_\mu u^\mu$ where $\Pcoh\Nat$
is the set of all $u\in\Realpto\Nat$ such that
$\sum_{n\in\Nat}u_n\leq 1$ (subprobability distributions on the
natural numbers). It is often convenient to use $\Fun S$ for
describing $S$ (no information is lost in doing so), what we do now.

Since $\PCOH$ (or rather the associated Kleisli CCC $\Kl\PCOH$) is a
model of probabilistic PCF, one has $\Fun S(u)\in\Intercc 01$ and one
can prove that $\Fun S(u)$ is exactly the probability that the
execution of $M$ converges if we apply it to a random integer
distributed along $u$ (such a random integer has also a probability
$1-\sum_{n=0}^\infty u_n$ to diverge): we call this property
\emph{adequacy} in the sequel. We can consider $S$ as a power series or
analytic function which can be infinite ``in width'' and ``in depth''.
\begin{itemize}
\item \emph{In width} because the $\mu$'s such that
  $\alpha_\mu\not=0$ can contain infinitely many different integers. A
  typical $S$ which is infinite in width is $S_1$ such that
  $\Fun{S_1}(u)=\sum_{n=0}^\infty u_n$ (the relevant $\mu$'s are the
  singleton multisets $\Mset n$, for all $n\in\Nat$).
\item \emph{In depth} in the sense that a given component of
  $u$ can be used an unbounded number of time. A typical $S$ which is
  infinite in depth is $S_2$ such that
  $\Fun{S_2}(u)=\sum_{n\in\Nat}\frac 1{2^{n+1}}u_0^n$.
\end{itemize}
More precisely, $S$ is finite in depth if the expression $\Fun S(u)$
has finitely non-zero terms when $u$ has a finite support. For
instance $S_3$ such that $\Fun{S_3}=\sum_{n=0}^\infty u_n u_0^n$ is
infinite in width but not in depth, in spite of the fact that it has
not a bounded degree in $u_0$. Notice that this notion of ``finiteness
in depth'' is at the core of the concept of finiteness space
introduced in~\cite{Ehrhard00b}. When $S$ arises as the semantics of a
term $M$ of probabilistic PCF (see~\cite{EhrhardPaganiTasson18} and
Section~\ref{sec:PCF-syntax}), it is generally not of finite depth
because $M$ can contain subterms $\Fix P$ representing recursive
definitions: remember that such a term ``reduces'' to $\App P{\Fix P}$.

Given a closed term $M$ of type $\Timpl\Tnat 1$ (or more generally a
closed term $M$ of type $\Timpl{\Tnat^k}1$ but we keep $k=1$ in this
introduction for readability) we are interested in approximating
effectively the probability that $\App MN$ converges when $N$ is a
closed term of type $\Tnat$ which represents a sub-distribution of
probabilities $u$, that is in approximating $\Fun S(u)$. To this end,
we try to find approximations of $S$ itself: then it will be enough to
apply these approximations to the subdistributions $u$'s we want to
consider. Finding approximations from below is not very difficult: it
suffices to consider terms $M_k$ obtained from $M$ by unfolding $k$
times all fixpoint operators it contains, that is replacing
hereditarily in $M$ each subterm of shape $\Fix P$ with the term
$\App P{\App P{\cdots\App P\Errdiv}}$ ($k$ occurrences of $P$ and
$\Errdiv$ is a constant which represents divergence and has semantics
$0$). The powerseries $S_k$ interpreting such an $M_k$ in $\PCOH$ is
then of finite depth\footnote{Because $M_k$ contains no
  fixpoints. This tree is a kind of PCF Böhm tree very similar to
  those considered in game semantics, \Eg~\cite{HylandOng00}.}  and can be
computed (\Eg~as a lazy data structure) by means of an adapted version
$\Kreval$ of the Krivine Machine, see Section~\ref{sec:Krivine-machine}. This
power series $S_k$ can still be an infinite object but since it is of
finite depth, by choosing a finite subset $J$ of $\Nat$, it is
possible to extract effectively from it a finite polynomial $S_k^J$ such
that $\Fun{S_k^J}(u)=\Fun{S_k}(u)$ when the support of $u$ is a subset
of $J$ (this extraction can be integrated in the Krivine Machine
itself, or applied afterwards to the lazy infinite data structure it
yields). The sequence $S_k$ is monotone in $\PCOH$ and has $S$ as lub
hence $\Fun{S_k}(u)$ is a monotone sequence in $\Intercc 01$ which
converges to $\Fun S(u)$. But there is no algorithm which, for any
closed term $P$ of type $1$ and any $p\in\Nat $ yields a $k$ such that
the semantics of $P_k$ is $2^{-p}$-close to that of $P$: such an
algorithm would make deciding almost-sure termination $\Pi^0_1$
whereas we know that this problem is $\Pi^0_2$-complete,
see~\cite{Katoen15}. So we cannot systematically know how good the
estimate $S_k$ is.

Nevertheless it would be nice to be able to approximate $\Fun S(u)$
\emph{from above} by some $\Fun{S^k}(u)$ where $S^k$ is again a
finite depth power series extracted from similar ``finite''
approximations $M^k$ of $M$: if we are lucky enough to find $k$ such
that $\Fun{S^k}(u)-\Fun{S_k}(u)\leq\epsilon$, we are sure that
$\Fun{S_k}(u)$ is an $\epsilon$-approximation of $\Fun S(u)$. This is
exactly what we do in this paper, developing first a denotational
account of these approximations.
% (just as PCSs account for approximation from below).

The PCS denotational account of approximations from below is based on
the fact that any PCS has a least element $0$ that we use to interpret
$\Errdiv$. For approximating from above we would need a maximal
element that we could use to interpret a constant $\Errconv$ to be
added to our PCF: we would then approximate $\Fix P$ with
$\App P{\App P{\cdots\App P\Errconv}}$. The problem is that, for a PCS
$X$, the associated domain $\Pcoh X$ (whose order relation is denoted
as $\leq$) has typically no maximal element; the PCS $\Snat$ of ``flat
integers'' whose web is $\Nat$ and $\Pcoh\Snat$ is the set of all
sub-probability distributions on $\Nat$, has no $\leq$-maximal element
since in this PCS $u\leq v$ means $\forall n\in\Nat\ u_n\leq v_n$.

So we consider PCSs equipped with an additional (pre)order relation
$\Exto{}$ for which such a maximal element can exist: an extensional
PCS is a tuple $X=(\Extc X,\Exto X,\Extd X)$, where $\Extc X$ is a PCS
(the carrier of $X$) and $\Exto X$ is a preorder relation on
$\Extd X$, the set of extensional elements of $X$, which is a subset
of $\Pcoh{\Extc X}$. These objects are a probabilistic analog of
Berry's bidomains~\cite{Berry78}. We prove that these objects form again
a model of classical linear logic $\EPCOH$ whose associated Kleisli
category $\Kl\EPCOH$ has fixpoint operators at all types. The main
features of this model are the following.
\begin{itemize}
\item At function type, $\Extd{\Timpl XY}$ is the set of all
  $s\in\Pcohp{\Timpl{\Extc X}{\Extc Y}}$ which are monotone wrt.~the
  extensional preorder, that is
  $\forall u,v\in\Extd X\ u\Exto Xv\Implies\Fun s(u)\Exto Y\Fun s(v)$
  (where $\Fun s:\Pcoh{\Extc X}\to\Pcoh{\Extc Y}$ is the ``stable''
  function associated with $s$) and, given $s,t\in\Extd{\Timpl XY}$,
  we stipulate that $s\Exto{\Timpl XY}t$ if
  $\forall u\in\Extd X\ \Fun s(u)\Exto Y\Fun t(u)$, that is
  $\Exto{\Timpl XY}$ is the extensional preorder.
\item There is an extensional PCS $\Flate\Nat$ whose web is
  $\Nat\cup\Eset\Nerr$ which is an extension of $\Snat$ in the sense
  that\footnote{$\One_\Nerr$ is the PCS
    $(\Eset\Nerr,\Eset\Nerr\times\Intercc 01)$.}
  $\Extc{\Flate\Nat}\Isom\Plus\Snat{\One_\Nerr}$ and
  $\Pcoh{\Extc{\Flate\Nat}}$ has an $\Exto{}$-maximal element,
  namely $\Base\Nerr$ (the $\Nat\cup\Eset\Nerr$-indexed family of
  scalars which maps $n\in\Nat$ to $0$ and $\Nerr$ to $1$).
\end{itemize}
Accordingly we extend probabilistic PCF
from~\cite{EhrhardPaganiTasson18} with two new constants
$\Errdiv,\Errconv$ of type $\Tnat$. The operational semantics of this
language is defined as a probabilistic rewriting system in the spirit
of~\cite{EhrhardPaganiTasson18}, with new rules for constants $\Errdiv$
and $\Errconv$ which are handled exactly in the same way, as error
exceptions. Then we define a syntactic preorder $\Termso$ on terms
such that $\Errdiv\Termso M\Termso\Errconv$ for all term of type
$\Tnat$, and such that
$M'\Termso M\Andprop N'\Termso \Fix M\Implies\App{M'}{N'}\Termso\Fix
M$ and
$M\Termso M'\Andprop\Fix M\Termso N'\Implies\Fix
M\Termso\App{M'}{N'}$. In particular, for any term $M$ we have
$M_k\Termso M\Termso M^k$ (where $M_k$ and $M^k$ are the ``finite''
approximations of $M$ obtained by unfolding all fixpoints $k$ times as
explained above starting from $\Errdiv$ and $\Errconv$
respectively\footnote{One can define error terms at all types by
  simply adding $\lambda$-abstractions in front of the ground type
  $\Errdiv$ and $\Errconv$.}).  We interpret\footnote{We omit the
  proof that the semantics is invariant by reduction and the proof of
  adequacy as they are simple adaptations of the corresponding proofs
  in~\cite{EhrhardPaganiTasson18}.} this language in $\EPCOH$ and prove
that this interpretation is extensionally monotone: if
$\Tseq{}{M,N}\sigma$ and $M\Termso N$ then
$\Psem M{}\Exto{\Tsem\sigma}\Psem N{}$ (where
$\Psem M{}\in\Extd{\Tsem\sigma}$ is the interpretation of the term $M$
in the interpretation of its type $\sigma$).

We adapt our approximation problem to this extension of PCF, without
changing its nature: assuming $\Tseq{x:\Tnat}M\Tnat$ and
$u\in\Pcoh\Snat$, approximate from above and below the probability $p$
that $\Subst MNx$ reduces to $\Errconv$, knowing that the probability
subdistribution of $N$ is $u$. To address it, we extend the Krivine
Machine $\Kreval$ to handle $\Errconv$. From a term $P$ \emph{without
  fixpoints} and such that $\Tseq{x:\Tnat}{P}{\Tnat}$, $\Kreval$
produces a (generally infinite) ``Böhm tree'' of which we extract a
power series $S$ of finite depth which coincides with the denotational
interpretation of $P$ in $\EPCOH$, or more precisely with the
$\Nerr$-component of this interpretation. So if $P\Termso M$ we have
$\Fun S(u)\leq p$ and if $M\Termso P$ then $p\leq\Fun S(u)$ by the
above monotonicity property and adequacy of the
semantics\footnote{Which guarantees that $p$ is equal to the
  interpretation of $M$ applied to $u$.}. This will hold in particular
if $P=M_k$ or $P=M^k$ respectively. Notice that if $J\subseteq\Nat$ is
finite and if the support of $u$ is a subset of $J$ then computing
$\Fun S(u)$ involves only a finite set of monomials, computable from
$J$.

\section*{Notations}
If $I$ is a set, we use $\Mfin I$ for the set of finite multisets of
elements of $I$, which are functions $\mu:I\to\Nat$ such that the set
$\Supp\mu=\Eset{i\in I\St\mu(i)\not=0}$ is finite. We use
$\Mset{\List i1k}$ for the multiset $\mu$ such that $\mu(i)$ is the
number of indices $l$ such that $i_l=i$. We use $\Mset{}$ for the
empty multiset and $+$ for the sum of multisets.

We use $\Realp$ for the set of real numbers $r$ such that $r\geq 0$
and we set $\Realpc=\Realp\cup\Eset{\infty}$ (the complete half-real
line).

If $i\in I$, we use $\Base i$ for the element of $\Realpto I$ such
that $(\Base i)_j=\Kronecker ij$ (the Kronecker symbol).

\section{Probabilistic coherence spaces (PCS)}\label{sec:PCS}

For the general theory of PCSs we refer
to~\cite{DanosEhrhard08,EhrhardPaganiTasson18}. We recall briefly the
basic definitions for the sake of self-containedness.

% \subsection{Basic definitions on PCSs}\label{sec:basics-PCSs}

Given an at most countable set $I$ and $u,u'\in\Realpcto I$, we set
$\Eval u{u'}=\sum_{i\in I}u_iu'_i\in\Realpc$. Given
$\cP\subseteq\Realpcto I$, we define $\Orth\cP\subseteq\Realpcto I$ as
\begin{align*}
  \Orth\cP=\{u'\in\Realpcto I\St\forall u\in\cP\ \Eval u{u'}\leq 1\}\,.
\end{align*}
Observe that if $\cP$ satisfies
\( \forall a\in I\,\exists u\in\cP\ u_a>0 \) and
\( \forall a\in I\,\exists m\in\Realp \forall u\in\cP\ u_a\leq m \)
then $\Orth\cP\in\Realpto I$ and $\Orth\cP$ satisfies the same two
properties.

A probabilistic pre-coherence space (pre-PCS) is a pair
$X=(\Web X,\Pcoh X)$ where $\Web X$ is an at most countable
set\footnote{This restriction is not technically necessary, but very
  meaningful from a philosophic point of view; the non countable case
  should be handled via measurable spaces and then one has to consider
  more general objects as in~\cite{EhrhardPaganiTasson18} for
  instance.} and $\Pcoh X\subseteq\Realpcto{\Web X}$ satisfies
$\Biorth{\Pcoh X}=\Pcoh X$. A probabilistic coherence space (PCS) is a
pre-PCS $X$ such that
% $\Pcoh X\subseteq\Realpto{\Web X}$ and
   %    \begin{align*}
\(
\forall a\in\Web X\,\exists u\in\Pcoh X\ u_a>0
\) and
\(
\forall a\in\Web X\,\exists m\in\Realp \forall u\in\Pcoh X\ u_a\leq m
\)
% \end{align*}
or equivalently
\( \forall a\in\Web X\quad0<\sup_{u\in\Pcoh X}u_a<\infty \) so that
$\Pcoh X\subseteq\Realpto{\Web X}$.  We define a "norm"
$\Norm\__X:\Pcoh X\to\Intercc 01$ by
$\Normsp x{\Pcohc X}=\inf\{r>0\St x\in r\,\Pcoh X\}$ that we shall use
for describing the coproduct of PCSs. Then
$\Orth X=(\Web X,\Orth{\Pcoh X})$ is also a PCS and $\Biorth X=X$.

Equipped with the order relation $\leq$ defined by $u\leq v$ if
$\forall a\in\Web X\ u_a\leq v_a$, any PCS $X$ is a complete partial
order (all directed lubs exist) with $0$ as least element. In general
this cpo is not a lattice.

Given $t\in\Realpcto{I\times J}$ considered as a matrix (where $I$ and
$J$ are at most countable sets) and $u\in\Realpcto I$, we define
$\Matappa tu\in\Realpcto J$ by $(\Matappa tu)_j=\sum_{i\in I}t_{i,j}u_i$
(usual formula for applying a matrix to a vector), and if
$s\in\Realpcto{J\times K}$ we define the product
$\Matapp st\in\Realpcto{I\times K}$ of the matrix $s$ and $t$ as usual
by $(\Matapp st)_{i,k}=\sum_{j\in J}t_{i,j}s_{j,k}$. This is an
associative operation.

Let $X$ and $Y$ be PCSs, a morphism from $X$ to $Y$ is a matrix
$t\in\Realpto{\Web X\times\Web Y}$ such that
$\forall u\in\Pcoh X\ \Matappa tu\in\Pcoh Y$. It is clear that the
identity matrix is a morphism from $X$ to $X$ and that the matrix
product of two morphisms is a morphism and therefore, PCS equipped
with this notion of morphism form a category $\PCOH$. There is a PCS
$\Limpl XY$ such that $\Web{\Limpl XY}=\Web X\times\Web Y$ and
$\Pcohp{\Limpl XY}$ is exactly the set of these matrices. Given any
$a$, we define $\One_a$ as the PCS whose web is $\Eset a$ and
$\Pcoh{\One_a}$ is $\Intercc 01$ or, pedantically,
$\Eset a\times\Intercc 01$. We write $\One$ instead of $\One_a$ if $a$
is a given element $*$, fixed once and for all.

The condition $t\in\PCOH(X,Y)=\Pcohp{\Limpl XY}$ is equivalent to
% \begin{align*}
\( \forall u\in\Pcoh X\,\forall v'\in\Pcoh{\Orth Y}\ \Eval{\Matappa
  tu}{v'}\leq 1 \)
%\end{align*}
and we have $\Eval{\Matappa tu}{v'}=\Eval t{\Tens u{v'}}$ where
$(\Tens u{v'})_{(a,b)}=u_av'_b$. Given PCS $X$ and $Y$ we define a PCS
$\Tens XY=\Orthp{\Limpl X{\Orth Y}}$ such that
$\Pcohp{\Tens XY}=\Biorth{\{\Tens uv\St w\in\Pcoh X\text{ and
  }v\in\Pcoh Y\}}$ where $\Tensp uv_{a,b}=u_av_b$. Equipped with this
operation $\ITens$ and the unit $\One$, $\PCOH$ is a symmetric
monoidal category (SMC) with isomorphisms of associativity
$\Assoc\in\PCOH(\Tens{\Tensp XY}{Z},\Tens X{\Tensp YZ})$, symmetry
$\Sym\in\PCOH(\Tens XY,\Tens YX)$, neutrality
$\Leftu\in\PCOH(\Tens\One X,X)$ and $\Rightu\in\PCOH(\Tens X\One,X)$
defined in the obvious way. This SMC $\PCOH$ is closed, with internal
hom of $X$ and $Y$ the pair $(\Limpl XY,\Evlin)$ where
$\Evlin\in\PCOH(\Tens{\Limplp XY}{X},Y)$ is given by
$\Evlin_{((a,b),a'),b'}=\Kronecker a{a'}\Kronecker b{b'}$ so that
$\Matappa\Evlin{\Tensp tu}=\Matappa tu$. This SMCC is *-autonomous
wrt.~the dualizing object $\Llbot=\One$ (essentially because
$\Limpl X\Llbot\Isom\Orth X$).

The following property is quite easy and very useful (actually we
already used it).
\begin{lemma}\label{lemma:matapp-eval}
  Let $t\in\Pcohp{\Limpl XY}$, $u\in\Pcoh X$ and
  $v'\in\Pcoh{\Orth Y}$. Then
  $\Eval{\Matappa tx}{v'}=\Eval t{\Tens u{v'}}=\Eval u{\Matappa{\Orth
      t}{v'}}$.
\end{lemma}

$\PCOH$ is cartesian: if $(X_i)_{i\in I}$ is an at most countable
family of PCSs, then $(\Bwith_{i\in I}X_i,(\Proj i)_{i\in I})$ is the
cartesian product of the $X_i$s, with
$\Web{\Bwith_{i\in I}X_i}=\Bunion_{i\in I}\{i\}\times\Web{X_i}$,
$(\Proj i)_{(j,a),a'}=\Kronecker ij\Kronecker a{a'}$, and
$u\in\Pcohp{\Bwith_{i\in I}X_i}$ if $\Matappa{\Proj i}u\in\Pcoh{X_i}$
for each $i\in I$ (for $u\in\Realpto{\Web{\Bwith_{i\in I}X_i}}$). It
is important to observe that $\Pcohp{\Bwith_{i\in I}X_i}$ is order
isomorphic to $\prod_{i\in I}\Pcoh{X_i}$ by the map
$u\mapsto(\Matappa{\Proj i}u)_{i\in I}$. Given
$\Vect u=(u(i))_{i\in I}\in\prod_{i\in I}\Pcoh{X_i}$ we use
$\Tuple{\Vect u}$ for the corresponding element of
$\Pcohp{\Bwith_{i\in I}X_i}$: $\Tuple{\Vect u}_{i,a}=u(i)_a$ for
$i\in I$ and $a\in\Web{X_i}$. The terminal object (which corresponds
to the case $I=\emptyset$) is the PCS $(\emptyset,\Eset 0)$.

Given $t(i)\in\PCOH(Y,X_i)$ for each $i\in I$, the unique morphism
$t=\Tuple{t(i)}_{i\in I}\in\PCOH(Y,\Bwith_{i\in I}X_i)$ such that
$\Proj i\Compl t=t_i$ is simply defined by
$t_{b,(i,a)}=(t_i)_{b,a}$. The dual operation $\Bplus_{i\in I}X_i$,
which is a coproduct, is characterized by
$\Web{\Bplus_{i\in I}X_i}=\Bunion_{i\in I}\{i\}\times\Web{X_i}$ and
$u\in\Pcohp{\Bplus_{i\in I}X_i}$ if $u\in\Pcohp{\Bwith_{i\in I}X_i}$
and $\sum_{i\in I}\Norm{\Matappa{\Proj i}{u}}_{X_i}\leq 1$. A
particular case is $\Snat=\Bplus_{n\in\Nat}X_n$ where $X_n=\One$ for
each $n$. So that $\Web\Snat=\Nat$ and $u\in\Realpto\Nat$ belongs to
$\Pcoh\Snat$ if $\sum_{n\in\Nat}u_n\leq 1$ (that is, $u$ is a
sub-probability distribution on $\Nat$). There are successor and
predecessor morphisms $\Ssuc,\Spred\in\PCOH(\Snat,\Snat)$ given by
$\Ssuc_{n,n'}=\Kronecker{n+1}{n'}$ and $\Spred_{n,n'}=1$ if $n=n'=0$
or $n=n'+1$ (and $\Spred_{n,n'}=0$ in all other cases). An element of
$\PCOH(\Snat,\Snat)$ is a (sub)stochastic matrix and the very idea of
this model is to represent programs as transformations of this kind,
and their generalizations.

As to the exponentials, one sets $\Web{\Excl X}=\Mfin{\Web X}$ and
$\Pcohp{\Excl X}=\Biorth{\{\Prom u\St u\in\Pcoh X\}}$ where, given
$\mu\in\Mfin{\Web X}$,
$\Prom u_\mu=u^\mu=\prod_{a\in\Web X}u_a^{\mu(a)}$. Then given
$t\in\PCOH(X,Y)$, one defines $\Excl t\in\PCOH(\Excl X,\Excl Y)$ in
such a way that $\Matappa{\Excl t}{\Prom x}=\Prom{(\Matappa tx)}$ (the
precise definition is not relevant here; it is completely determined
by this equation). There are natural transformations
$\Der X\in\PCOH(\Excl X,X)$ and
$\Digg X\in\PCOH(\Excl X,\Excl{\Excl X})$ which are fully
characterized by $\Matappa{\Der X}{\Prom u}=u$ and
$\Matappa{\Digg X}{\Prom u}=\Promm u$ which equip $\Excl\_$ with a
comonad structure. There are also Seely isomorphisms
$\Seelyz\in\PCOH(\One,\Excl\Top)$ and
$\Seelyt_{X,Y}\PCOH(\Tens{\Excl X}{\Excl Y},\Exclp{X\IWith Y})$ which
equip this comonad with a strong monoidal structure from de SMC
$(\PCOH,\Top,\mathord{\IWith})$ to the SMC
$(\PCOH,\One,\mathord{\ITens})$. They are fully characterized by the
equations $\Matappa{\Seelyz}{\Prom 0}=\Base *$ and
$\Matappa{\Seelyt}{\Tensp{\Prom u}{\Prom v}}=\Prom{\Tuple{u,v}}$.

Using these structures one can equip any object $\Excl X$ with a
commutative comonoid structure consisting of a weakening morphism
$\Weak X\in\PCOH(\Excl X,\One)$ and a contraction morphism
$\Contr X\in\PCOH(\Excl X,\Tens{\Excl X}{\Excl X})$ characterized by
$\Matappa{\Weak X}{u}=\Base *$ and
$\Matappa{\Contr X}{\Prom u}=\Tens{\Prom u}{\Prom u}$.

The resulting cartesian closed category\footnote{This is the Kleisli
  category of ``$\oc$'' which has actually a comonad structure that we
  do not make explicit here, again we refer
  to~\cite{DanosEhrhard08,EhrhardPaganiTasson18}.}  $\Kl\PCOH$ can be
seen as a category of functions (actually, of stable functions as
proved in~\cite{Crubille18}). Indeed, a morphism
$t\in\Kl\PCOH(X,Y)=\PCOH(\Excl X,Y)=\Pcohp{\Limpl{\Excl X}{Y}}$ is
completely characterized by the associated function
$\Fun t:\Pcoh X\to\Pcoh Y$ such that
$\Fun t(u)=\Matappa t{\Prom u}=\left(\sum_{\mu\in\Web{\Excl
      X}}t_{\mu,b}u^\mu\right)_{b\in\Web Y}$ so that we consider
morphisms as power series. They are in particular monotonic and Scott
continuous functions $\Pcoh X\to\Pcoh Y$. In this cartesian closed
category, the product of a family $(X_i)_{i\in I}$ is
$\Bwith_{i\in I}X_i$ (written $X^I$ if $X_i=X$ for all $i$), which is
compatible with our viewpoint on morphisms as functions since
$\Pcohp{\Bwith_{i\in I}X_i}=\prod_{i\in I}\Pcoh{X_i}$ up to trivial
iso. The object of morphisms from $X$ to $Y$ is $\Limpl{\Excl X}{Y}$
with evaluation mapping
$(t,u)\in\Pcohp{\Limpl{\Excl X}{Y}}\times\Pcoh X$ to $\Fun t(u)$.
% that we simply denote as $t(x)$ from now on. 
The well defined function $\Pcohp{\Limpl{\Excl X}X}\to\Pcoh X$ which
maps $t$ to $\sup_{n\in\Nat}t^n(0)$ is a morphism of $\Kl\PCOH$ (and
thus can be described as a power series in
$t=(t_{\mu,a})_{\mu\in\Mfin{\Web X},a\in\Web X}$) by standard categorical
considerations using cartesian closeness: it provides us with fixpoint
operators at all types.
% \footnote{It is quite a valuable miracle that the
%   operator which maps a power series in $\Pcohp{\Limpl{\Excl X}X}$ to
%   its least fixed point is itself a power series.}.

Given $t\in\PCOH(X,Y)$, we also use $\Fun t$ for the associated
function $\Pcoh X\to\Pcoh Y$. More generally if for instance
$t\in\Pcohp{\Tens{\Excl X}{Y},Z}$, we use $\Fun t$ for the associated
function $\Pcoh X\times\Pcoh Y\to\Pcoh Z$, given by
$\Fun t(u,v)=\Matappa t{\Tensp{\Prom u}{v}}$. This function fully
characterizes $t$ (that is the mapping $t\mapsto\Fun t$ is injective);
this can be seen by considering
$\Curlin(t)\in\PCOH(\Excl X,\Limpl YZ)$.

\section{Extensional PCS}

Let $X$ be a PCS. A \emph{pre-extensional structure} on $X$ is a pair
$\cU=(\cE,\mathord\Extleq)$ where $\cE\subseteq\Pcoh X$ and $\Extleq$
is a binary relation on $\cE$. We define then the dual pre-extensional
structure $\Orth\cU=(\cE',\mathord{\Extleq'})$ on $\Orth X$ as
follows\footnote{Notice the kind of role swapping between $\cE$ and
  $\Extleq$ in this definition; this justifies our choice of
  presenting these structures as pairs $(\cE,\Extleq)$ and not simply
  as relations $\Extleq$ on $\Pcoh X$.}:
\begin{itemize}
\item if $u'\in\Pcoh{\Orth X}$, one has $u'\in\cE'$ iff
  $\forall u,v\in\cE\ u\Extleq v\Implies \Eval u{u'}\leq\Eval v{u'}$
\item and, given $u',v'\in\cE'$, one has $u'\Extleq' v'$ iff
  $\forall u\in\cE\ \Eval u{u'}\leq\Eval u{v'}$.
\end{itemize}
Let $\cU_1=(\cE_1,\mathord{\Extleq_1})$ and
$\cU_2=(\cE_2,\mathord{\Extleq_2})$ be pre-extensional structures on
$X$, we write $\cU_1\subseteq\cU_2$ if $\cE_1\subseteq\cE_2$ and
$\mathord{\Extleq_1}\subseteq\mathord{\Extleq_2}$.

\begin{lemma}\label{lemma:ext-struct-dual-monotone}
  If $\cU_1\subseteq\cU_2$ then
  $\Orth{\cU_2}\subseteq\Orth{\cU_1}$. One has
  $\cU\subseteq\Biorth\cU$ and therefore $\Orth\cU=\Triorth\cU$.
\end{lemma}

Given a pre-extensional structure $(\cE,\Extleq)$, when we write
$u\Extleq v$, we always assume implicitly that $u,v\in\cE$.  An
\emph{extensional structure} on $\Pcoh X$ is a pre-extensional
structure $\cU$ such that $\cU=\Biorth\cU$.
\begin{proposition}\label{prop:pre-ext-basic-props}
  If $\cU=(\cE,\mathord\Extleq)$ is an extensional structure on the
  PCS $X$, then $\Extleq$ is a transitive and reflexive (that is, a
  preorder) relation on $\cE$. Moreover,
  $\forall u,v\in\cE\ u\leq v\Implies u\Extleq v$, $0\in\cE$, $0$ is
  $\Extleq$-minimal, $\cE$ and $\Extleq$ are sub-convex, closed under
  multiplication by scalars in $[0,1]$ and closed under lubs of
  $\leq$-increasing $\omega$-chains. Concerning $\Extleq$, this means
  that
  \begin{itemize}
  \item if $I$ is a finite set, $\lambda_i\in\Realp$ for each $i\in I$
    with $\sum_{i\in I}\lambda_i\leq 1$, $u(i)\Extleq v(i)$ for
    each $i\in I$ then
    $\sum_{i\in I}\lambda_iu(i)\Extleq\sum_{i\in I}\lambda_iv(i)$
  \item and if $(u(n))_{n\in\Nat}$ and $(v(n))_{n\in\Nat}$ are
    $\leq$-monotone in $\Pcoh X$ and such that $u(n)\Extleq v(n)$ for
    each $n\in\Nat$, then
    $\sup_{n\in\Nat}u(n)\Extleq\sup_{n\in\Nat}v(n)$ (this can be
    generalized to directed families).
  \end{itemize}
\end{proposition}

An \emph{extensional PCS} is a triple $X=(\Extc X,\Extd X,\Exto X)$
where $\Extc X$ is a PCS (the carrier) and $(\Extd X,\Exto X)$ is an
extensional structure on $\Extc X$. The dual of $X$ is then
$\Orth X=(\Orth{\Extc X},\Orth{(\Extd X,\Exto X)})$, so that
$\Biorth X=X$ by definition.  An extensional PCS $X$ is
\emph{discrete} if $\Extd X=\Pcoh X$ and $u\Exto Xv$ iff $u\leq v$ (in
$\Pcoh X$). Observe that if $X$ is discrete then $\Orth X$ is also
discrete. Of course any PCS can be turned into an extensional PCS by
endowing it with its discrete extensional structure.

\subsection{Extensional PCS as a model of Linear Logic}
\begin{lemma}\label{lemma:ext-struct-limpl-charact}
  Let $X$ and $Y$ be extensional PCS, we define a pre-extensional
  structure $(\cE,\mathord\Extleq)$ on the PCS
  $\Limpl{\Extc X}{\Extc Y}$ by:
  \begin{itemize}
  \item given $t\in\Pcohp{\Limpl{\Extc X}{\Extc Y}}$, one has
    $t\in\cE$ if $\forall u\in\Extd X\ \Matappa tu\in \Extd Y$ and
    $\forall u(1),u(2)\in\Extd X\ u(1)\Exto Xu(2)\Implies\Matappa
    t{u(1)}\Exto Y\Matappa t{u(2)}$
  \item and given $t(1),t(2)\in\Pcohp{\Limpl{\Extc X}{\Extc Y}}$, one
    has $t(1)\Extleq t(2)$ if
    $\forall u\in\Extd X\ \Matappa{t(1)}{u}\Exto Y\Matappa{t(2)}{u}$.
  \end{itemize}
  We denote as $\Limpl XY$ the extensional PCS defined by
  $\Extc{\Limpl XY}=\Limpl{\Extc X}{\Extc Y}$, $\Extd{\Limpl XY}=\cE$
  and $\mathord{\Exto{\Limpl XY}}=\mathord\Extleq$.
\end{lemma}
\begin{proof}
  Let $(\cE',\Extleq')$ be the pre-extensional structure on
  $\Orthp{\Limpl{\Extc X}{\Extc Y}}=\Tens{\Extc X}{\Orth{\Extc Y}}$
  defined by
  \begin{itemize}
  \item
    $\cE'=\Eset{\Tens u{v'}\St u\in\Extd X\text{ and }v'\in\Extd{\Orth
        Y}}$
  \item
    $\Extleq'=\Eset{(\Tens{u(1)}{v'(1)},\Tens{u(2)}{v'(2)})\St
      u(1)\Exto X u(2)\text{ and }v'(1)\Exto{\Orth Y}v'(2)}$.
  \end{itemize}
  We prove that $(\cE,\mathord\Extleq)=\Orth{(\cE',\Extleq')}$.

  Let $(\cE'',\Extleq'')=\Orth{(\cE',\Extleq')}$ which is an
  extensional structure on $\Limpl{\Extc X}{\Extc Y}$.

  Let $t\in\cE$ and let us prove that $t\in\cE''$. So let
  $u(1)\Exto X u(2)$ and $v'(1)\Exto{\Orth Y}v'(2)$, we have
  \begin{align*}
    \Eval t{\Tens{u(1)}{v'(1)}}=\Eval{\Matappa
    t{u(1)}}{v'(1)}\leq\Eval{\Matappa t{u(1)}}{v'(2)} 
  \end{align*}
  since $\Matappa t{u(1)}\in\Extd Y$, and we have
  $\Eval{\Matappa t{u(1)}}{v'(2)}\leq\Eval{\Matappa t{u(2)}}{v'(2)}$
  because $\Matappa t{u(1)}\Exto Y\Matappa t{u(2)}$. Conversely let
  $t\in\cE''$ and let us prove that $t\in\cE$. First, let
  $u\in\Extd X$, we must prove that $\Matappa tu\in\Extd Y$. So let
  $v'(1)\Exto{\Orth Y}v'(2)$, we must prove that
  $\Eval{\Matappa t{u}}{v'(1)}\leq\Eval{\Matappa t{u}}{v'(2)}$, which
  results from our assumption on $t$ and from
  Lemma~\ref{lemma:matapp-eval}. Next we must prove that, if
  $u(1)\Exto Xu(2)$, then $\Matappa t{u(1)}\Exto Y\Matappa t{u(2)}$. So
  let $y'\in\Extd{\Orth Y}$, we must show that
  $\Eval{\Matappa t{u(1)}}{y'}\leq\Eval{\Matappa t{u(2)}}{y'}$ which
  results from Lemma~\ref{lemma:matapp-eval} and from the assumption
  that $t\in\cE''$.

  Let now $t(1)\Extleq t(2)$ and let us prove that
  $t(1)\Extleq'' t(2)$.  So let $u\in\Extd X$ and
  $v'\in\Extd{\Orth Y}$, we must prove that
  $\Eval{t(1)}{\Tens u{v'}}\leq \Eval{t(2)}{\Tens u{v'}}$ which again
  results from Lemma~\ref{lemma:matapp-eval} and from the fact that
  $\Matappa{t(1)}u\Extleq\Matappa{t(2)}u$. Conversely assume that
  $t(1)\Extleq'' t(2)$ and let us prove that $t(1)\Extleq t(2)$. So
  let $u\in\Extd X$, we must prove that
  $\Matappa{t(1)}{u}\Exto Y\Matappa{t(2)}{u}$. So let
  $v'\in\Extd{\Orth Y}$, we must prove
  $\Eval{\Matappa{t(1)}{u}}{v'}\leq\Eval{\Matappa{t(2)}{x}}{y'}$ which
  again results from Lemma~\ref{lemma:matapp-eval}.
\end{proof}

From this definition it results that if $s\in\Extd{\Limpl XY}$ and
$t\in\Extd{\Limpl YZ}$, one has $\Matapp ts\in\Extd{\Limpl XZ}$, and
also that $\Id_{\Extc X}\in\Extd{\Limpl XX}$.  So we have defined a
category $\EPCOH$ whose objects are the extensional PCS and where
$\EPCOH(X,Y)=\Extd{\Limpl XY}$, identities and composition being
defined as in $\PCOH$.
  
\begin{remark}
  It is very important to notice that a morphism $t\in\EPCOH(X,Y)$
  acts on \emph{all} the elements of $\Pcoh{\Extc X}$ and not only on
  those which belong to $\Extd X$ (intuitively, the extensional
  elements).
\end{remark}

\begin{lemma}\label{lemma:traspose-in-epcoh}
  Let $t\in\Pcohp{\Limpl XY}$, we have $t\in\EPCOH(X,Y)$ iff
  $\Orth t\in\EPCOH(\Orth Y,\Orth X)$. Let $t(1),t(2)\in\EPCOH(X,Y)$,
  one has $t(1)\Exto{\Limpl XY}t(2)$ iff
  $\Orth{t(1)}\Exto{\Limpl{\Orth Y}{\Orth X}}\Orth{t(2)}$.
\end{lemma}

\subsubsection{Multiplicative structure}
As usual we set $\Tens XY=\Orthp{\Limpl X{\Orth Y}}$ so that
$\Extc{\Tens XY}=\Tens{\Extc X}{\Extc Y}$.

\begin{lemma}\label{lemma:tens-extensional}
  If $u\in\Extd X$ and $v\in\Extd Y$, we have
  $\Tens uv\in\Extd{\Tens XY}$. And if $u(1)\Exto Xu(2)$ and
  $v(1)\Exto Yv(2)$ then
  $\Tens{u(1)}{v(1)}\Exto{\Tens{X}{Y}}\Tens{u(2)}{v(2)}$.
\end{lemma}

For proving the categorical properties of a $\Tens{}{}$ operation
defined in that way, one always starts with proving a lemma
characterizing \emph{bilinear} morphisms.

\begin{lemma}\label{lemma:bilinear-epcs-charact}
  Let $t\in\PCOH(\Tens{\Extc X}{\Extc Y},\Extc Z)$. One has
  $t\in\EPCOH(\Tens XY,Z)$ iff the following conditions hold:
  \begin{enumerate}
  \item if $u\in\Extd X$ and $v\in\Extd Y$ then
    $\Matappa t{\Tensp uv}\in\Extd Z$\label{it:bilinear-epcs-charact-1}
  \item if $u(1)\Exto X u(2)$ and $v(1)\Exto Y v(2)$ then
    $\Matappa t{\Tensp{u(1)}{v(1)}}\Exto Z\Matappa
    t{\Tensp{u(2)}{v(2)}}$.\label{it:bilinear-epcs-charact-2}
  \end{enumerate}
  Let $t(1),t(2)\in\PCOH(\Tens{\Extc X}{\Extc Y},\Extc Z)$, one has
  $t(1)\Exto{\Limpl{\Tens{\Extc X}{\Extc Y}}Z}t(2)$ iff for all
  $u\in\Extd X$ and $v\in\Extd Y$, one has
  $\Matappa{t(1)}{\Tensp uv}\Exto Z\Matappa{t(2)}{\Tensp uv}$.
\end{lemma}

\begin{lemma}\label{lemma:epcs-morph-tensor}
  Let $t(i)\in\EPCOH(X_i,Y_i)$ for $i=1,2$, then $\Tens{t(1)}{t(2)}$,
  which is an element of
  $\PCOH(\Tens{\Extc{X_1}}{\Extc{X_2}},\Tens{\Extc{Y_1}}{\Extc{Y_2}})$,
  satisfies
  $\Tens{t(1)}{t(2)}\in\EPCOH(\Tens{X_1}{X_2},\Tens{Y_1}{Y_2})$.
\end{lemma}
\Beginproof
We apply Lemma~\ref{lemma:bilinear-epcs-charact}, so let first
$u(i)\in\Extd{X_i}$ for $i=1,2$. We have
$\Matappa{\Tensp{t(1)}{t(2)}}{\Tensp{u(1)}{u(2)}}
=\Tens{(\Matappa{t(1)}{u(1)})}{(\Matappa{t(2)}{u(2)})}
\in\Extd{\Tens{Y_1}{Y_2}}$
by Lemma~\ref{lemma:tens-extensional} since we have
$\Matappa{t(i)}{u(i)}\in\Extd{Y_i}$ for $i=1,2$.
Next assume that $u^1(i)\Exto{X_i}u^2(i)$ for $i=1,2$. We have
$\Matappa{t(i)}{u^1(i)}\Exto{X_i}\Matappa{t(i)}{u^2(i)}$ for $i=1,2$ and
hence
$\Tens{(\Matappa{t(1)}{u^1(1)})}{(\Matappa{t(2)}{u^2(1)})}
\Exto{\Tens{Y_1}{Y_2}}\Tens{(\Matappa{t(1)}{u^1(2)})}{(\Matappa{t(2)}{u^2(2)})}$
by Lemma~\ref{lemma:tens-extensional}.
\Endproof

Let $X$, $Y$ and $Z$ be
extensional PCS, we have an isomorphism
\begin{align*}
  \Assoc:\PCOH(\Tens{\Tensp{\Extc X}{\Extc Y}}{\Extc Z}, \Tens{\Extc
  X}{\Tensp{\Extc Y}{\Extc Z}}) 
\end{align*}
defined in the obvious way:
$\Assoc_{((a,b),c),(a',(b',c'))}
=\Kronecker{a}{a'}\Kronecker{b}{b'}\Kronecker{c}{c'}$.
We prove first that
$\Assoc\in\EPCOH(\Tens{\Tensp XY}{Z},\Tens X{\Tensp YZ})$. For this,
by Lemma~\ref{lemma:traspose-in-epcoh}, it suffices to prove that
$\Orth\Assoc\in\EPCOH(\Orthp{\Tens X{\Tensp YZ}},\Orthp{\Tens{\Tensp
    XY}{Z}})$, that is
\begin{align*}
  \Orth\Assoc\in\EPCOH(\Limpl X{\Limplp Y{\Orth Z}},\Limpl{\Tens
  XY}{\Orth Z})\,.
\end{align*}
Let $t\in\Extd{\Limpl X{\Limplp Y{\Orth Z}}}$, we prove that
$\Matappa{\Orth\Assoc}t\in\Extd{\Limpl{\Tens XY}{\Orth Z}}$
applying Lemma~\ref{lemma:bilinear-epcs-charact}.
Let first $u\in\Extd X$ and $v\in\Extd Y$, we have
$\Matappa{(\Matappa{\Orth\Assoc}t)}{\Tensp uv}=\Matappa{(\Matappa
  tu)}v\in\Extd{\Orth Z}$.
Next let $u(1)\Exto X u(2)$ and $v(1)\Exto Y v(2)$, we have
$\Matappa t{u(1)}\Exto{\Limpl Y{\Orth Z}}\Matappa t{u(2)}$, hence
$\Matappa{(\Matappa t{u(1)})}{v(1)} \Exto{\Orth Z}\Matappa{(\Matappa
  t{u(2)})}{v(1)}$ and since
$\Matappa t{u(2)}\in\Extd{\Limpl{Y}{\Orth Z}}$, we have
\begin{align*}
  \Matappa{(\Matappa t{u(2)})}{v(1)}
  \Exto{\Orth Z}\Matappa{(\Matappa t{u(2)})}{v(1)}
  \Exto{\Orth Z}\Matappa{(\Matappa t{u(2)})}{v(2)}
\end{align*}
so that
$\Matappa{(\Matappa t{u(2)})}{v(1)} \Exto{\Orth Z}\Matappa{(\Matappa
  t{u(2)})}{v(2)}$ by transitivity of $\Exto{\Orth Z}$
(Proposition~\ref{prop:pre-ext-basic-props}).

Next, given $t(1)\Exto{\Limpl X{\Limplp Y{\Orth Z}}}t(2)$, we prove
that
$\Matappa{\Orth\Assoc}{t(1)}\Exto{\Limpl{\Tens XY}{\Orth
    Z}}\Matappa{\Orth\Assoc}{t(2)}$. By
Lemma~\ref{lemma:bilinear-epcs-charact} it suffices to prove that,
given $u\in\Extd X$ and $v\in\Extd Y$, one has
$\Matappa{(\Matappa{\Orth\Assoc}{t(1)})}{\Tensp uv} \Exto{\Orth
  Z}\Matappa{(\Matappa{\Orth\Assoc}{t(2)})}{\Tensp uv}$ which results
from the definition of $\Assoc$ which yields
$\Matappa{(\Matappa{\Orth\Assoc}{t(i)})}{\Tensp
  uv}=\Matappa{\Matappa{(t(i)}{u})}{v}$ and from our assumption on the
$t(i)$'s.

Let $\beta=\Funinv\Assoc$, we prove that
$\beta\in\EPCOH(\Tens{X}{\Tensp{Y}{Z}},\Tens{\Tensp{X}{Y}}{Z})$,
that is
\begin{align*}
  \Orth\beta\in\EPCOH(\Limpl{\Tens XY}{\Orth Z},\Limpl X{\Limplp
  Y{\Orth Z}})\,.
\end{align*}
So let $t\in\Extd{\Limpl{\Tens XY}{\Orth Z}}$, we
prove first that
$\Matappa{\Orth\beta}{t}\in\Extd{\Limpl X{\Limplp Y{\Orth Z}}}$.
Let $u\in\Extd X$ and $v\in\Extd Y$, we have
$\Matappa{(\Matappa{(\Matapp{\Orth\beta}{t})}u)}{v}=\Matappa t{\Tensp
  uv}\in\Extd{\Orth Z}$. This shows that
$\Matappa{(\Matappa{\Orth\beta}{t})}u\in\Extd{\Limpl Y{\Orth Z}}$.
% Let $x\in\Extd X$, we must prove that
% $\Matapp{(\Matapp{\Orth\beta}{t})}x\in\Extd{\Limpl Y{\Orth Z}}$. Let
% first $y\in\Extd Y$, we have
% $\Matapp{(\Matapp{(\Matapp{\Orth\beta}{t})}x)}{y}=\Matapp t{\Tensp
%   xy}\in\Extd{\Orth Z}$.
%
Next let $u(1)\Exto Xu(2)$ and $v\in\Extd Y$, we have
$\Matappa{(\Matappa{(\Matapp{\Orth\beta}{t})}{u(i)})}{v}=\Matapp
t{\Tensp{u(i)}{v}}$ and hence
$\Matappa{(\Matappa{(\Matappa{\Orth\beta}{t})}{u(1)})}{v} \Exto{\Orth
  Z}\Matappa{(\Matappa{(\Matappa{\Orth\beta}{t})}{u(2)})}{v}$ since
$\Tens{u(1)}{v}\Exto{\Tens XY}\Tens{u(2)}{v}$.  This shows that
$\Matappa{(\Matappa{\Orth\beta}{t})}{u(1)} \Exto{\Limpl{Y}{\Orth
    Z}}\Matappa{(\Matappa{\Orth\beta}{t})}{u(2)}$.
% Next let $x(1)\Exto Xx(2)$, we prove that
% $\Matapp{(\Matapp{\Orth\beta}{t})}{x(1)} \Exto{\Limpl{Y}{\Orth
%     Z}}\Matapp{(\Matapp{\Orth\beta}{t})}{x(2)}$, so let $y\in\Extd Y$,
% we must prove that
% $\Matapp{(\Matapp{(\Matapp{\Orth\beta}{t})}{x(1)})}{y} \Exto{\Orth
%   Z}\Matapp{(\Matapp{(\Matapp{\Orth\beta}{t})}{x(2)})}{y}$ which again
% results from the fact that
% $\Matapp{(\Matapp{(\Matapp{\Orth\beta}{t})}{x(i)})}{y}=\Matapp
% t{\Tensp{x(i)}{y}}$.
%
Now let $u(1)\Exto X u(2)$ and $v\in\Extd Y$, we prove similarly
$\Matappa{(\Matappa{(\Matappa{\Orth\beta}{t})}{u(1)})}{v}\Exto{\Orth
  Z}\Matappa{(\Matappa{(\Matappa{\Orth\beta}{t})}{u(1)})}{v}$ which shows
that
$\Matappa{(\Matappa{\Orth\beta}{t})}{u(1)} \Exto{\Limpl Y{\Orth
    Z}}\Matappa{(\Matappa{\Orth\beta}{t})}{u(2)}$.
% Now let $x(1)\Exto X x(2)$, we must prove that
% $\Matapp{(\Matapp{\Orth\beta}{t})}{x(1)} \Exto{\Limpl Y{\Orth
% Z}}\Matapp{(\Matapp{\Orth\beta}{t})}{x(2)}$. So let $y\in\Extd Y$,
% we must prove
% $\Matapp{(\Matapp{(\Matapp{\Orth\beta}{t})}{x(1)})}{y}\Exto{\Orth
% Z}\Matapp{(\Matapp{(\Matapp{\Orth\beta}{t})}{x(1)})}{y}$ which is
% obtained in the same manner.

Last let $t(1)\Exto{\Limpl{\Tens XY}{Z}}t(2)$, $u\in\Extd X$ and
$v\in\Extd Y$, we have $\Tens uv\in\Extd{\Tens XY}$ and hence
$\Matappa{t(1)}{\Tensp uv}\Exto{\Orth Z}\Matappa{t(2)}{\Tensp uv}$ that
is
$\Matappa{(\Matappa{(\Matappa{\Orth\beta}{t(1)})}u)}v \Exto{\Orth
  Z}\Matappa{(\Matappa{(\Matappa{\Orth\beta}{t(2)})}u)}v$, so
$\Matappa{(\Matappa{\Orth\beta}{t(1)})}u\Exto{\Limpl Y{\Orth
    Z}}\Matappa{(\Matappa{\Orth\beta}{t(2)})}v$ and hence
$\Matappa{\Orth\beta}{t(1)}\Exto{\Limpl X{\Limplp Y{\Orth
      Z}}}\Matappa{\Orth\beta}{t(2)}$.

The fact that we have a symmetry isomorphism
$\Sym\in\EPCOH(\Tens XY,\Tens YX)$ such that
$\Sym_{(a,b),(b',a')}=\Kronecker a{a'}\Kronecker b{b'}$ is a direct
consequence of Lemmas~\ref{lemma:bilinear-epcs-charact}
and~\ref{lemma:tens-extensional}.

The tensor unit is $\One$, equipped with the discrete extensional
structure. It is straightforward to check that $\Leftu$ and $\Rightu$
are isos in $\EPCOH$. To summarize, we have proven the first part of
the following result.

\begin{theorem}\label{th:EPCOH-star-aut}
  Equipped with $(\mathord\otimes,\One,\Leftu,\Rightu,\Assoc,\Sym)$,
  the category $\EPCOH$ is symmetric monoidal. Moreover, this
  symmetric monoidal category is closed (SMCC), the object of linear
  morphisms $X\to Y$ being $(\Limpl XY,\Evlin)$\footnote{Where
    $\Evlin\in\PCOH(\Tens{\Limplp{\Extc X}{\Extc Y}}{\Extc X},\Extc
    Y)$ is the evaluation morphism of $\PCOH$, see
    Section~\ref{sec:PCS}.}. This SMCC has a *-autonomous structure,
  with $\One$ as dualizing object.
\end{theorem}
Given $t\in\EPCOH(\Tens ZX,Y)$, we use $\Curlin(t)$ for the curryfied
version of $t$ which belongs to $\EPCOH(Z,\Limpl XY)$ and is defined
exactly as in $\PCOH$ (that is $\Curlin(t)_{c,(a,b)}=t_{(c,a),b}$).

\subsubsection{The additives}

The additive structure is quite simple. Given an at most countable
family $(X_i)_{i\in I}$ of extensional PCS, we define a pre-extensional PCS
$X=\Bwith_{i\in I}X_i$ as follows:
\begin{itemize}
\item $\Extc X=\Bwith_{i\in I}\Extc{X_i}$ (in $\PCOH$ of course)
% \item
%   $\Extd X=\Eset{\Tuple{u(i)}_{i\in I}\St\forall i\in I\
%     u(i)\in\Extd{X_i}}$
\item
  $\Extd X=\Eset{u\in\Pcoh X\St\forall i\in I\
    \Matappa{\Proj i}u\in\Extd{X_i}}$
% \item $\Tuple{u(i)(1)}_{i\in I}\Exto X\Tuple{u(i)(2)}_{i\in I}$ if
%   $\forall i\in I\ u(i)(1)\Exto{X_i}u(i)(2)$.
\item $u\Exto Xv$ if
  $\forall i\in I\ \Matappa{\Proj i}u\Exto{X_i}\Matappa{\Proj i}v$.
\end{itemize}

\begin{lemma}\label{lemma:epcs-plus-charact}
  Let $X=\Bwith_{i\in I}X_i$. Then the extensional PCS $\Orth X$ is
  characterized by
  \begin{itemize}
  \item $\Extc{\Orth X}=\Bplus_{i\in I}\Orth{\Extc{X_i}}$
  % \item
  %   $\Extd{\Orth X}=\Eset{\Tuple{u^i}_{i\in I}\in\Pcoh{\Orth{\Extc
  %         X}}\St\forall i\in I\ u^i\in\Extd{\Orth{X_i}}}$
  % \item $\Tuple{u^i(1)}_{i\in I}\Exto{\Orth X}\Tuple{u^i(2)}_{i\in I}$
  %   if $\forall i\in I\ u^i(1)\Exto{\Orth{X_i}}u^i(2)$.
  \item
    $\Extd{\Orth X}=\Eset{u'\in\Pcoh{\Orth{\Extc
          X}}\St\forall i\in I\ \Matappa{\Proj i}{u'}\in\Extd{\Orth{X_i}}}$
  \item $u'\Exto{\Orth X}v'$
    if $\forall i\in I\ \Matappa{\Proj i}{u'}
    \Exto{\Orth{X_i}}\Matappa{\Proj i}{v'}$.
  \end{itemize}
\end{lemma}
Remember that the elements of $\Pcoh{\Extc{\Orth X}}$ are the
$u'\in\prod_{i\in I}\Pcoh{\Orth{\Extc{X_i}}}$
such that $\sum_{i\in I}\Norm{\Matappa{\Proj i}{u'}}_{\Orth{X_i}}\leq 1$. 

\begin{lemma}\label{lemma:ext-PCS-cart-prod}
  The pre-extensional PCS $X=\Bwith_{i\in I}X_i$ is an extensional
  PCS. For each $i\in I$, the projection
  $\Proj j\in\PCOH(\Extc X,\Extc{X_i})$ belongs to $\EPCOH(X,X_i)$
  and, equipped with these projections, $X$ is the cartesian product of
  the $X_i$'s in $\EPCOH$.
\end{lemma}

If all $X_i$'s are the same extensional PCS $X$, we use the notation
$X^I$ for the extensional PCS $\Bwith_{i\in I}X_i$. Notice that
$\Web{\Extc{X^I}}=I\times\Web{\Extc X}$.

\begin{remark}
  One should observe that the constructions introduced so far preserve
  discreteness. More precisely, if $X$ and $Y$ are discrete extensional
  PCSs, so are $\Limpl XY$, $\Tens XY$, $\Orth X$ etc, and if
  $(X_i)_{i\in I}$ is an at most countable family of discrete
  extensional PCSs, so are $\Bwith_{i\in I}X_i$ and
  $\Bplus_{i\in I}X_i$. This is not the case of the exponentials.
\end{remark}

\subsubsection{The exponentials}
Let $X$ be an extensional PCS. We define
$\Excl X=(\Excl{\Extc X},\Biorth{(\Extd{\Excl X}^0,\Exto{\Excl
    X}^0)})$ where the pre-extensional structure
$(\Extd{\Excl X}^0,\Exto{\Excl X}^0)$ on $\Excl{\Extc X}$ is defined
by:
\begin{align*}
  \Extd{\Excl X}^0&=\Eset{\Prom u\St u\in\Extd X}\\
  \mathord{\Exto{\Excl X}^0}&=\Eset{(\Prom{u(1)},\Prom{u(2)})\St
  u(1)\Exto X u(2)}\,.
\end{align*}
The main consequence of this definition is the following.
\begin{theorem}\label{th:exp-ext-charact}
  Let $X$ and $Y$ be extensional PCSs.
  \begin{itemize}
  \item Given $t\in\Pcoh(\Limpl{\Excl{\Extc X}}{\Extc Y})$, one has
    $t\in\Extd{\Limpl{\Excl X}{Y}}$ iff
    $\forall u\in\Extd X\ \Matappa t{\Prom u}\in\Extd Y$ and for any
    $u(1)\Exto X u(2)$, one has
    $\Matappa t{\Prom{u(1)}}\Exto Y\Matappa t{\Prom{u(2)}}$.
  \item Given $t(1),t(2)\in\Extd{\Limpl{\Excl X}{Y}}$, one has
    $t(1)\Exto{\Limpl{\Excl X}{Y}}t(2)$ iff for all $u\in\Extd X$, one
    has $\Matappa{t(1)}{\Prom u}\Exto Y\Matappa{t(2)}{\Prom u}$.
  \end{itemize}
\end{theorem}

Now we can derive important consequences of that result.
\begin{theorem}\label{th:epcs-excl-comonad}
  Given $t\in\EPCOH(X,Y)$, one has
  $\Excl t\in\EPCOH(\Excl X,\Excl Y)$. Moreover,
  $\Der{\Extc X}\in\EPCOH(\Excl X,X)$ and
  $\Digg{\Extc X}\in\EPCOH(\Excl X,\Excl{\Excl X})$.
\end{theorem}

Of course the diagram commutations which hold in $\PCOH$ relative to
this constructs still hold in $\EPCOH$ (composition is identical in
both categories) and so $\Excl\_$ equipped with these two natural
transformation is a comonad.

The Seely isomorphisms of $\PCOH$ are easily checked to be morphisms
in $\EPCOH$. The case of $\Seelyz$ is straightforward so we deal only with
$\Seelyt\in\PCOH(\Tens{\Excl{\Extc X}}{\Excl{\Extc
    Y}},\Exclp{\With{\Extc X}{\Extc Y}})$, which is the isomorphism in
$\PCOH$ given by
\begin{align*}
  \Seelyt_{(\Mset{\List a1m},\Mset{\List b1n}),\rho}
  =
  \begin{cases}
    1 &\text{if }\rho=\Mset{(1,a_1),\dots,(1,a_m),(2,b_1),\dots,(2,b_n)}\\
    0 &\text{otherwise}
  \end{cases}
\end{align*}
which satisfies
\(
  \Matapp{\Seelyt}{\Tensp{\Prom u}{\Prom v}}=\Prom{\Tuple{u,v}}
\).
We know that
$\Curlin{\Seelyt}\in\PCOH(\Excl{\Extc X},\Limpl{\Excl{\Extc
    Y}}{\Exclp{\With{\Extc X}{\Extc Y}}})$ and we prove that actually
$t=\Curlin{\Seelyt}\in\EPCOH(\Excl X,\Limpl{\Excl Y}{\Exclp{\With
    XY}})$, applying again Theorem~\ref{th:exp-ext-charact}.

Let $u\in\Extd X$, we prove that
$\Matappa t{\Prom u\in\Extd{\Limpl{\Excl Y}{\Exclp{\With XY}}}}$.
So let $v\in\Extd Y$, we have
$\Matappa{(\Matappa t{\Prom u})}{\Prom
  v}=\Prom{\Tuple{u,v}}\in\Extd{\Exclp{\With XY}}$ because
$\Tuple{u,v}\in\Extd{\With XY}$. Let $v(1)\Exto Yv(2)$, we have
$\Matappa{(\Matappa t{\Prom u})}{\Prom{v(1)}}
=\Prom{\Tuple{u,v(1)}}\Exto{\Exclp{\With
    XY}}\Prom{\Tuple{u,v(2)}}=\Matappa{(\Matappa t{\Prom
    u})}{\Prom{v(2)}}$ because
$\Tuple{u,v(1)}\Exto{\With XY}\Tuple{u,v(1)}$ by definition of
$\With XY$.
Next let $u(1)\Exto X u(2)$ and let use check that
$\Matappa t{\Prom{u(1)}} \Exto{\Limpl{\Excl Y}{\Exclp{\With
      XY}}}\Matappa t{\Prom{u(2)}}$. So let $v\in\Extd Y$, it suffices
to prove that
$\Matappa{(\Matappa t{\Prom{u(1)}})}{\Prom v} \Exto{\Exclp{\With XY}}
\Matappa{(\Matappa t{\Prom{u(2)}})}{\Prom v}$ which is obtained exactly
as above.

Last we have also to prove that
$\Funinv{(\Seelyt)}\in\EPCOH(\Exclp{\With XY},\Tens{\Excl X}{\Excl
  Y})$; this is an easy consequence of
Theorem~\ref{th:exp-ext-charact}, of the fact that
$\Extd{\With XY}=\Extd X\times\Extd Y$ (up to the isomorphism between
$\Pcohp{\With XY}$ and $\Pcoh X\times\Pcoh Y$) and similarly,
$\Exto{\With XY}$ is the product preorder and 
$\Matappa{\Funinv{(\Seelyt)}}{\Prom{\Tuple{u,v}}}=\Tens{\Prom u}{\Prom
  v}$.

This ends the proof that $\EPCOH$ is a (new) Seely category
(see~\cite{Mellies09}), that is, a categorical model of classical linear
logic.

As a consequence, the Kleisli category $\Kl\EPCOH$ of $\Excl\_$ over
$\EPCOH$ is cartesian closed.

Given $t\in\Kl\PCOH(\Extc X,\Extc Y)=\PCOH(\Excl{\Extc X},\Extc Y)$,
we use $\Fun t$ for the corresponding function
$\Pcoh{\Extc X}\to\Pcoh{\Extc Y}$ defined by
$\Fun t(u)=\Matapp t{\Prom u}$. Such a $t$ is a morphism of
$\Kl\EPCOH$ iff $\forall u\in\Extd X\ \Fun t(u)\in\Extd Y$. Moreover
such a $t$ is $\Extleq$-monotone, that is
\begin{align*}
  u\Exto X v\Implies\Klapp tu\Exto Y\Klapp tv
\end{align*}
and also, given $s,t\in\Kl\EPCOH(X,Y)$, one has
\begin{align*}
  s\Exto{\Limpl{\Excl X}{Y}}t\Equiv
  \forall u\in\Extd X\ \Klapp su\Exto Y\Klapp tu
\end{align*}
that is, $\Exto{\Limpl{\Excl X}{Y}}$ is the pointwise (pre)order on
functions.

We use $\Simpl XY$ for the object of morphisms from $X$ to $Y$ in that
category, which is $\Limpl{\Excl X}{Y}$, and is equipped with the
evaluation morphism $\Ev\in\Kl\EPCOH(\With{(\Impl XY)}{X},Y)$
characterized of course by $\Fun\Ev(s,u)=\Fun s(u)$.

\begin{example}\label{ex:decreasing-ext-sequence}
  Let $X$ be the extensional PCS $\Simpl\One\One$. Up to a trivial
  iso, an element of $\Pcoh{\Extc X}$ is a family $s\in\Realpto\Nat$
  such that $\sum_{n\in\Nat}s_n\leq 1$ and the associated function
  $\Fun s:\Intercc 01\to\Intercc 01$ is given by
  $\Fun s(u)=\sum_{n\in\Nat}s_nu^n$. Since $\One$ is discrete, one has
  $\Extd X=\Pcoh{\Extc X}$, and given $s,t\in\Extd X$, one has
  \begin{align*}
    s\Exto Xt\Equiv\forall u\in\Intercc 01\ \Fun s(u)\leq\Fun s(t)\,.
  \end{align*}
  For each $n\in\Nat$, let $e(n)$ be the element of $\Pcoh{\Extc X}$
  such that $e(n)_i=\Kronecker ni$, that is $\Fun{e(n)}(u)=u^n$, one
  has $\forall n\in\Nat\ e(n+1)\Exto Xe(n)$ since
  $\forall u\in\Intercc 01\ u^{n+1}\leq u^n$. So $(e(n))_{n\in\Nat}$
  is an $\Exto X$-decreasing $\omega$-chain which has $0$ as
  glb. Notice that this glb is not the pointwise glb of the $e(n)$'s
  considered as functions since $\inf_{n\in\Nat}\Fun{e(n)}(1)=1$
  (glb.~computed in $\Intercc 01$). Observe also that the $e(n)$'s are
  pairwise unbounded for the standard $\leq$ order on
  $\Pcoh{\Extc X}$ (each of them is $\leq$-maximal in this PCS).
\end{example}

\subsubsection{$\EPCOH$ is an enriched category over preorders}
\label{sec:pcoh-ext-enriched-model}
Given objects $X$ and $Y$ of $\EPCOH$, we can equip $\EPCOH(X,Y)$ with
the preorder relation $\Exto{\Limpl XY}$ that we denote as
$\Exto{X,Y}$. In other words, given $s,t\in\EPCOH(X,Y)$, one has
$s\Exto{X,Y}t$ iff
$\forall u\in\Extd X\ \Matappa su\Exto Y\Matappa tu$. This turns
$\EPCOH$ into an enriched category over the monoidal category of
partial order (with the usual product of preorders as monoidal
product) and actually into an ``enriched Seely category'' in the sense
that all the constructions involved in the definition of $\EPCOH$ as a
Seely category are $\Exto{}$-monotone. For instance, if
$s,t\in\EPCOH(X,Y)$ satisfy $s\Exto{X,Y}t$, then
$\Excl s\Exto{\Excl X,\Excl Y}\Excl t$ as an easy consequence of
Theorem~\ref{th:exp-ext-charact}.

\subsubsection{General recursion, fixpoints}

\begin{theorem}\label{th:ext-pcoh-fixpoint}
  Let $t\in\Kl\EPCOH(X,X)$, the least fixpoint of $t$ in $\Pcoh X$,
  which is $\sup_{n\in\Nat}\Fun t^n(0)$, belongs to $\Extd X$. As a
  consequence, the $\PCOH$ least fixpoint operator
  $\Fixpcoh{\Extc Y}\in\PCOH(\Impl{\Extc Y}{\Extc Y},\Extc Y)$
  (characterized by $\Fun{\Fixpcoh Y}(t)=\sup_{n\in\Nat}\Fun t^n(0)$)
  actually belongs to $\Kl\EPCOH(\Impl YY,Y)$.
\end{theorem}
\Beginproof
For the first statement, remember that $0\in\Extd X$ by
Proposition~\ref{prop:pre-ext-basic-props} and hence by a
straightforward induction $\forall n\in\Nat\ \Fun t^n(0)\in\Extd
X$. Therefore, by Proposition~\ref{prop:pre-ext-basic-props} again, we
have $\sup_{n\in\Nat}\Fun t^n(0)\in\Extd X$.

The second part of the theorem is proven by applying the first part in
the following special case: $X=(\Simpl{(\Simpl YY)}{Y})$ and
$t\in\Kl\EPCOH(X,X)$ is characterized by
$\Fun{\Fun t(F)}(s)=\Fun s(\Fun F(\Fun s))$ for $F\in\Pcoh{X}$ and
$s\in\Pcohp{\Simpl YY}$. The existence of $t$ results from the
cartesian closeness of $\Kl\EPCOH$. Then the least fixed point of
$\Fun t$ is $\Fixpcoh{\Extc Y}$ and this prove our contention by the
first part of the theorem.
\Endproof

\begin{example}\label{ex:unit-upper-fixpoint}
  Let again $X$ be the extensional PCS $\Simpl\One\One$.  If we are
  given $F\in\Extd{\Simpl XX}$, we know that $\Fun F$ has a least
  fixpoint $t\in\Extd X$ given by $t=\sup_{n\in\Nat}\Fun
  F^n(0)$. Given $u\in\Intercc 01$, we know that
  $\Fun t(u)\in\Intercc 01$, and if we set $t(n)=\Fun F^n(0)$, we have
  $t(n)\leq t$ and hence $\Fun t(u)\geq\Fun{t(n)}(u)$ and hence
  $\Fun{t(n)}(u)$ gives us a lower approximation of $\Fun t(u)$. We
  even know that $\Fun t(u)$ is the lub of these approximations but
  this gives us no clue on how good a given approximation
  $\Fun{t(n)}(u)$ is (how far it is from the target value
  $\Fun t (u)$).

  One main feature of the $\Exto X$ relation is that it has a maximal
  element, namely $e(0)$ (notations of
  Example~\ref{ex:decreasing-ext-sequence}) that we simply denote as
  $1$, since it represents the constant function $1$, in sharp
  contrast with the $\leq$ relation on $\Pcoh{\Extc X}$. Since
  $F\in\Extd{\Simpl XX}$, the function $\Fun F$ is
  $\Exto X$-monotonic, and hence $(s(n)=\Fun F^n(1))_{n\in\Nat}$ is an
  $\Exto X$-decreasing sequence. Therefore the sequence
  $(\Fun{s(n)}(u))_{n\in\Nat}$ is decreasing in $\Intercc 01$ (for the
  usual order relation which coincides with $\Exto\One$). Moreover,
  since $0\Exto\One 1$, we have
  $\forall n\in\Nat\ \Fun F^n(0)\Exto X\Fun F^n(1)$ by induction on
  $n$ and hence for all $n\in\Nat$, $\Fun t(u)\leq\Fun{s(n)}(u)$. In
  particular, given $\epsilon>0$, if we find $n\in\Nat$ such that
  $\Fun{t(n)}(u)-\Fun{s(n)}(u)\leq\epsilon$, we are certain that
  $\Fun{t(n)}(u)$ and $\Fun{s(n)}(u)$ are at most at distance
  $\epsilon$ of the probability $\Fun t(u)$ we are interested in.

  It may happen that for some $\epsilon>0$ the condition
  $\Fun{t(n)}(u)-\Fun{s(n)}(u)\leq\epsilon$ never holds. Take for
  instance $F$ to be the identity and $u=1$: in that case
  $\forall n\in\Nat\ \Fun{t(n)}(1)=0\text{ and }\Fun{s(n)}=1$. But if
  we manage to fulfill this condition by taking $n$ big enough, we are
  certain to get and $\epsilon$-approximation whereas having only the
  lower approximations $t(n)$ we could never know, whatever be the value
  of $n$. Moreover we can expect that some reasonable syntactic
  guardedness restrictions on the program defining $F$ will guarantee
  that these $\epsilon$-approximations always exist (such restrictions
  certainly already exist in the rich literature on abstract
  interpretation).
\end{example}

The remainder of the paper is essentially devoted to extending this
idea to more useful datatypes.

\subsection{Flat types with errors}

Another crucial feature of extensional PCSs is that they allow to
build basic data types extended with an ``error'' or ``escape''
element which is maximal for the $\Exto{}$ preorder but not for the
$\leq$ preorder, just as $0$ is minimal, thus allowing to extend the
observation of Example~\ref{ex:unit-upper-fixpoint} to languages
having datatypes like booleans or integers and not just the poorly
expressive unit type.

Given an at most countable set $I$, we defined the ordinary PCS
$\Flat I$ as follows: $\Web{\Flat I}=I$ and
$\Pcohp{\Flat I}=\Eset{u\in\Realpto I\St\sum_{i\in I}u_i\leq 1}$. As
an extensional PCS, it is equipped with the discrete structure.

\subsubsection{General definitions and basic properties}
\label{sec:flate-general}
Now we introduce another object $\Flate I$, with a non-trivial
extensional structure. First we take
$\Extc{\Flate I}=\Plus{\Flat I}{\One_\Nerr}$.

% Now we have to explain how we will interpret the ground type of
% integers in this model. The main idea is to extend the usual
% interpretation by the object $\Snat=(\Nat,\Pcoh\Snat)$ where
% $\Pcoh\Snat=\Eset{x\in\Realpto\Snat\St \sum x_n\leq 1}$ so as to
% obtain an object of $\EPCOH$ where the $\Extleq$ has a maximal
% element.

% So we define the object $\Extnat$ as follows:
% $\Extc\Extnat=\Plus\Snat\One$ (in $\PCOH$) so that
% $\Web\Extnat=\Nat\cup\Eset{\Nerr}$ and $x\in\Pcoh{\Extc\Extnat}$ if
% $\sum_{n\in\Nat}x_n+x_\Nerr\leq 1$.
%
We take $\Extd{\Flate I}=\Pcohp{\Extc{\Flate I}}$ and we are left with
defining $\Exto{\Flate I}$. Given $u,v\in\Pcohp{\Extc{\Flate I}}$ we set
\begin{align*}
  \Invi uv=\Eset{i\in I\St u_i>v_i}
\end{align*}
that is, $\Invi uv$ is the set of all $i\in I$ ``where it is not true
that $u$ is less than $v$'' (the inversion indices) and we stipulate
that $u\Exto{\Flate I} v$ if
\begin{align*}
  \sum_{i\in\Invi uv}(u_i-v_i)\leq v_\Nerr-u_\Nerr\,,
  \text{\quad that is\quad}
  u_\Nerr+\sum_{i\in\Invi uv}u_i\leq v_\Nerr+\sum_{i\in\Invi uv}v_i\,.
\end{align*}
Notice that this condition implies that $u_\Nerr\leq v_\Nerr$.

In other words $u\Exto{\Flate I} v$ means that the difference of
probabilities of the ``error'' $\Nerr$ compensates the sum of all
probability inversions from $u$ to $v$.  In that way we have equipped
the PCS $\Extc{\Flate I}$ with an extensional structure, as we
prove now.
\begin{proposition}
  We have $(\Extd{\Flate I},\Exto{\Flate I})=\Orth{\cU'}$ where
  $\cU'=(\cE',\Extleq')$ is the pre-extensional structure on
  $\Orth{\Extc{\Flate I}}$ defined as follows:
  \begin{itemize}
  \item
    $\cE'=\Eset{u'\in\Pcohp{\Orth{\Extc{\Flate I}}}\St \forall i\in I\
      u'_i\leq u'_\Nerr}$
  \item and $u'\Extleq' v'$ if $u'\leq v'$.
  \end{itemize}
  Therefore $\Flate I=(\Extc{\Flate I},\Exto{\Flate I})$ is an
  extensional PCS. One has
  $\forall u\in\Extd{\Flate I}\ u\Exto{\Flate I}\Base\Nerr$ and it is
  also true that $\Orthp{\Flate I}=(\Orth{\Extc{\Flate I}},\cU')$.
\end{proposition}
\Beginproof
Using these notations, let $(\cE,\Extleq)=\Orth{\cU'}$. We have
$\cE=\Extd{\Flate I}=\Pcoh{\Extc{\Flate I}}$ since, if
$u\in\Pcoh{\Extc{\Flate I}}$ and $u'(1)\Extleq'u'(2)$ one has
$u'(1)\leq u'(2)$ and hence $\Eval u{u'(1)}\leq\Eval u{u'(2)}$.
Assume now that $u(1)\Extleq u(2)$ and let us prove that
$u(1)\Exto{\Flate I} u(2)$. Let $u'\in\Pcoh{\Orthp{\Flate\Nat}}$ be defined
by
\begin{align*}
  u'_a=
  \begin{cases}
    1&\text{if }a\in I \text{ and }a\in\Invi{u(1)}{u(2)}\\
    1&\text{if }a=\Nerr\\
    0&\text{otherwise}
  \end{cases}
\end{align*}
then we have $u'\in\cE'$ and hence
$\Eval{u(1)}{u'}\leq\Eval{u(2)}{u'}$ which means exactly that
$u(1)\Exto{\Flate I} u(2)$.
Conversely assume that $u(1)\Exto{\Flate I} u(2)$ and let $u'\in\cE'$,
we have
\begin{align*}
  \Eval{u(1)}{u'}
  =\sum_{i\in I}u(1)_iu'_i+u(1)_\Nerr u'_\Nerr
  &=\sum_{\Biind{i\in I}{i\notin\Invi{u(1)}{u(2)}}}u(1)_iu'_i
    +\sum_{\Biind{i\in I}{i\in\Invi{u(1)}{u(2)}}}u(1)_iu'_i
    +u(1)_\Nerr u'_\Nerr\\
  &\leq\sum_{\Biind{i\in I}{i\notin\Invi{u(1)}{u(2)}}}u(2)_iu'_i
    +\sum_{\Biind{i\in I}{i\in\Invi{u(1)}{u(2)}}}u(1)_iu'_i
    +u(1)_\Nerr u'_\Nerr\\
  &=\sum_{i\in\Nat}u(2)_iu'_i
    +\sum_{\Biind{i\in I}{i\in\Invi{u(1)}{u(2)}}}(u(1)_i-u(2)_i)u'_i
    +u(1)_\Nerr u'_\Nerr\\
  &=\Eval{u(2)}{u'}
    +\sum_{\Biind{i\in I}{i\in\Invi{u(1)}{u(2)}}}(u(1)_i-u(2)_i)u'_i
    +(u(1)_\Nerr-u(2)_\Nerr) u'_\Nerr\\
  &\leq\Eval{u(2)}{u'}
    +\sum_{\Biind{i\in I}{i\in\Invi{u(1)}{u(2)}}}(u(1)_i-u(2)_i)u'_\Nerr
    +(u(1)_\Nerr-u(2)_\Nerr) u'_\Nerr
\end{align*}
because $u(1)_i-u(2)_i\geq 0$ when $i\in\Invi{u(1)}{u(2)}$ and
$u'_i\leq u'_\Nerr$ for each $i\in I$, which shows that
$\Eval{u(1)}{u'}\leq\Eval{u(2)}{u'}$
since $u(1)\Exto{\Flate I} u(2)$ which implies that
\begin{align*}
  \sum_{\Biind{i\in I}{i\in\Invi{u(1)}{u(2)}}}(u(1)_i-u(2)_i)u'_\Nerr
    +(u(1)_\Nerr-u(2)_\Nerr) u'_\Nerr\leq 0\,.
\end{align*}

Next we observe\footnote{This was actually the goal of all this
  construction!} that for all $u\in\Pcoh{\Extc{\Flate I}}$, one has
$u\Exto{\Flate I}\Base\Nerr$: we have
\begin{align*}
  \sum_{i\in\Invi x{\Base\Nerr}}(u_i-(\Base\Nerr)_i)
  =\sum_{i\in I}u_i\leq 1-u_\Nerr
\end{align*}
by definition of $\Pcoh{\Extc{\Flate I}}$ and this
is exactly the definition of $u\Exto{\Flate I}\Base\Nerr$.

Last let us prove that $\cU'=\Orth{(\Extd{\Flate I},\Exto{\Flate I})}$. The
direction $\cU'\subseteq\Orth{(\Extd{\Flate I},\Exto{\Flate I})}$ results
from what we have proven so far so let us prove the converse, and let
us introduce the notation $(\cE^\bullet,\Extleq^\bullet)$ for
$\Orth{(\Extd{\Flate I},\Exto{\Flate I})}$.

Let $u'\in\cE^\bullet$ and $i\in I$, we have
$\Base i\Exto{\Flate I}\Base\Nerr$ and hence
$\Eval{\Base i}{u'}\leq\Eval{\Base\Nerr}{u'}$ which proves
that $\forall i\in I\ u'_i\leq u'_\Nerr$, that is
$u'\in\cE'$.
Assume next that $u(1)'\Extleq^\bullet u(2)'$ and let
$a\in\Web{\Extc{\Flate I}}$, then since $\Base a\in\Extd{\Flate I}$ we have
$\Eval{\Base a}{u'(1)}\leq\Eval{\Base a}{u'(2)}$, that is
$u'(1)_a\leq u'(2)_a$ so that $u'(1)\leq u'(2)$.
\Endproof

\begin{example}
  The extensional PCS $\Flate\emptyset$ coincides with $\One$
  equipped with its discrete extensional structure.

  The elements of $\Extd{\Flate\One}$ are all pairs
  $u=(u_*,u_\Nerr)\in\Realpto 2$ such that $u_*+u_\Nerr\leq 1$, and
  $u\Exto{\Flate\One}v$ if $u_\Nerr\leq v_\Nerr$ and
  $u_*+u_\Nerr\leq v_*+v_\Nerr$; in this case we don't need to mention
  $\Invi uv$.

  Now let $X=\Flate{\Eset{0,1}}$ which represents the type of
  booleans, so an element of $\Extd X$ is a triple
  $u=(u_0,u_1,u_\Nerr)\in\Realpto 3$ such that
  $u_0+u_1+u_\Nerr\leq 1$. We have for instance
  $u=(1,0,0)\Exto Xv=(0,0,1)$ because in this case $\Invi uv=\Eset
  0$. Notice that we do not have for instance
  $u=(1,0,0)\Exto Xv=(0,1,0)$ in spite of the fact that
  $u_\Nerr\leq v_\Nerr$ and $u_0+u_1+u_\Nerr\leq v_0+v_1+v_\Nerr$. In
  this case we need to use the sets $\Invi{u}{v}=\Eset{0}$ to
  characterize $\Exto X$: we do not have $u_0+u_\Nerr\leq v_0+v_\Nerr$
  in this specific example.
\end{example}

We shall use the morphism $\Scatch I\in\EPCOH(\Flate I,\Llbot)$
given by $(\Scatch I)_{a,*}=\Kronecker a\Nerr$, in other
words $\Fun{\Scatch I}(u)=u_\Top$.

\subsubsection{Case construct with error}
\label{sec:pcoh-ext-case}
Given an extensional PCS $X$, remember that we use $X^I$ for the
extensional PCS $\Bwith_{i\in I}X_i$ where $X_i=X$ for each $i\in
I$. Let $J$ be another at most countable set. We define
\begin{align*}
  \Pcasee {I,J}\in\Realpto{
  \Web{\Extc{\Limpl{\Tens{\Flate I}{\Exclp{(\Flate J)^I}}}{\Flate J}}}}
\end{align*}
as follows:
\begin{align*}
  \Pcasee {I,J}_{a,\mu,b}=
  \begin{cases}
    1 & \text{if }a=i\in I
    \text{ and }b\in\Web{\Extc{\Flate J}}
    \text{ and }\mu=\Mset{(i,b)}\\
    1 & \text{if }a=\Nerr\text{ and }\mu=\Mset{}\text{ and }b=\Nerr\\
    0 & \text{otherwise}  
  \end{cases}
\end{align*}
Given $u\in\Pcohp{\Extc{\Flate I}}$ and
$\Vect v=(v(i))_{i\in I}\in\prod_{i\in I}\Pcohp{\Flate J}$, let
$w=\sum_{i\in I}u_iv(i)+u_\Nerr\Base\Nerr\in\Realpto{\Web{\Extc{\Flate
      J}}}$. We have
\begin{align*}
  \sum_{j\in J}w_j+w_\Nerr
  =\sum_{j\in J}\sum_{i\in I}u_iv(i)_j+\sum_{i\in I}u_iv(i)_\Nerr+u_\Nerr
  &=\sum_{i\in I}u_i\left(\sum_{j\in J}v(i)_j+v(i)_\Nerr\right)+u_\Nerr\\
  &\leq \sum_{i\in I}u_i+u_\Nerr\leq 1\,.
\end{align*}
This shows that
$\Pcasee{I,J}\in\Pcohp{\Extc{\Limpl{\Tens{\Flate I}{\Exclp{(\Flate
        J)^I}}}{\Flate J}}}$, the associated function
$\Funpcasee{I,J}:\Pcohp{\Extc{\Flate I}}\times\prod_{i\in
  I}\Pcohp{\Extc{\Flate J}}\to\Pcohp{\Extc{\Flate J}}$ being given by
\begin{align*}
  \Funpcasee{I,J}(u,\Vect v)=\sum_{i\in I}u_iv(i)+u_\Nerr\Base\Nerr
\end{align*}
where $\Vect v=(v(i))_{i\in I}$.
\begin{lemma}
  The morphism $\Pcasee{I,J}$ is extensional, that is, it belongs to
  $\EPCOH(\Tens{\Flate I}{\Exclp{(\Flate J)^I}},\Flate J)$.
\end{lemma}
\begin{proof}
  Let $u^1,u^2\in\Extd{\Flate I}$ be such that $u^1\Exto{\Flate I}u^2$
  and let $\Vect{v^1},\Vect{v^2}\in\prod_{i\in I}\Extd{\Flate J}$ be
  such that $v^1(i)\Exto{\Flate J}v^2(i)$ for each $i\in I$, we must
  prove that
  $\Funpcasee{I,J}(u^1,\Vect{v^1})\Exto{\Flate
    J}\Funpcasee{I,J}(u^2,\Vect{v^2})$. So let
  $w'\in\Extd{\Orthp{\Flate J}}$ which simply means that
  $w'\in\Intercc 01^{J\cup\Eset\Nerr}$ and
  $\forall j\in J\ w'_j\leq w'_\Nerr$. We have
  \begin{align*}
    \Eval{\Funpcasee{I,J}(u^1,\Vect{v^1})}{w'}
    &=\sum_{j\in J}\Funpcasee{I,J}(u^1,\Vect{v^1})_jw'_j
      +\Funpcasee{I,J}(u^1,\Vect{v^1})_\Nerr w'_\Nerr\\
    &=\sum_{j\in J}\sum_{i\in I}u^1_iv^1(i)_jw'_j
      +\sum_{i\in I}u^1_iv^1(i)_\Nerr w'_\Nerr+u^1_\Nerr w'_\Nerr\\
    &=\sum_{i\in I}u^1_i\Eval{v^1(i)}{w'}+u^1_\Nerr w'_\Nerr\\
    &\leq\sum_{i\in I}u^1_i\Eval{v^2(i)}{w'}+u^1_\Nerr w'_\Nerr\\
    &=\Eval{u^1}{u'}
  \end{align*}
  where $u'\in\Intercc 01^{I\cup\Eset\Nerr}$ is defined by
  $u'_i=\Eval{v^2(i)}{w'}$ and $u'_\Nerr=w'_\Nerr$. We have
  \begin{align*}
    u'_i
    =\sum_{j\in J}v^2(i)_jw'_j+v^2(i)_\Nerr w'_\Nerr
    \leq\sum_{j\in J}v^2(i)_jw'_\Nerr+v^2(i)_\Nerr w'_\Nerr
      \text{\quad since }\forall j\ w'_j\leq w'_\Nerr
    \leq w'_\Nerr=u'_\Nerr
  \end{align*}
  hence $u'_i\leq u'_\Nerr$ since
  $\sum_{j\in J}v^2(i)_j+v^2(i)_\Nerr\leq 1$; this shows that
  $u'\in\Extd{\Orthp{\Flate I}}$. It follows that
  $\Eval{u^1}{u'}\leq\Eval{u^2}{u'}$ since $u^1\Exto{\Flate
    I}u^2$. Therefore
  \begin{align*}
    \Eval{\Funpcasee{I,J}(u^1,\Vect{v^1})}{w'}
    \leq\Eval{u^1}{u'}
    \leq\Eval{u^2}{u'}
    =\Eval{\Funpcasee{I,J}(u^2,\Vect{v^2})}{w'}
  \end{align*}
  thus proving our contention.
\end{proof}

\subsubsection{The ``let'' with error}
\label{sec:pcoh-ext-let}

Let $X$ be an extensional PCS and let $I$ and $J$ be two sets which
are at most countable. Given
$t\in\EPCOH(\Tens{\Excl X}{\Excl{\Flate I}},\Flate J)$, we define
$\Slete t\in\Realpto{\Web{\Extc{\Limpl{\Tens{\Excl X}{\Flate
          I}}{\Flate J}}}}$ by
\begin{align*}
  \Slete t_{\rho,a,b}=
  \begin{cases}
    1 & \text{if }\rho=\Mset{}\text{ and }a=b=\Nerr\\
    \sum_{\mu\in\Mfin{\Eset i}}t_{\rho,\mu,b}
    & \text{if }a=i\in I
  \end{cases}
\end{align*}
The associated function
$\Fun s:\Pcoh{\Extc X}\times\Pcohp{\Extc{\Flate
    I}}\to\Realpcto{\Web{\Extc{\Flat J}}}$ is is characterized by
$\Fun s(u,v)=\sum_{i\in I}v_i\Fun t(u,\Base i)+v_\Nerr\Base\Nerr$. We have
\begin{align*}
  \sum_{j\in J}\Fun s(u,v)_j+\Fun s(u,v)_\Nerr
  &=\sum_{j\in J}\sum_{i\in I}v_i\Fun t(u,\Base i)_j
    +\sum_{i\in I}v_i\Fun t(u,\Base i)_\Nerr+v_\Nerr\\
  &=\sum_{i\in I}v_i\left(\sum_{j\in J}\Fun t(u,\Base i)_j
    +\Fun t (u,\Base i)_\Nerr\right)+v_\Nerr\\
  &\leq \sum_{i\in I}v_i+v_\Nerr\leq 1
\end{align*}
which shows that
$s\in\PCOH({\Tens{\Excl{\Extc X}}{\Extc{\Flate I}}},\Extc{\Flate
  J})$.

The next lemma uses the notations above.
\begin{lemma}\label{lemma:pcoh-extensional-let}
  Let $w'\in\Extd{\Orthp{\Flate J}}$. Then
  $\Eval{\Fun s(u,v)}{w'}=\Eval v{v'(t,u,w')}$ where
  $v'(t,u,w')\in\Extd{\Orthp{\Flate I}}$ is given by:
  $v'(t,u,w')_i=\Eval{\Fun t(u,\Base i)}{w'}$ for $i\in I$ and
  $v'(t,u,w')_\Nerr=w'_\Nerr$.  If
  $t(1)\Exto{\Tens{\Excl X}{\Excl{\Flate I}},\Flate J}t(2)$ and
  $u(1)\Exto{X}u(2)$ then
  $v'(t(1),u(1),w')\Exto{\Orthp{\Flate J}}v'(t(2),u(2),w')$.
\end{lemma}

Now we prove that $s\in\EPCOH({\Tens{\Excl{X}}{\Flate I}},\Flate
J)$. Since $\Extd{\Flate J}=\Pcohp{\Extc{\Flate J}}$ we only have to
prove $\Exto{}$-monotonicity. So let $u(1),u(2)\in\Extd X$ with
$u(1)\Exto Xu(2)$ and let $v(1),v(2)\in\Pcohp{\Extc{\Flate I}}$ with
$v(1)\Exto{\Flate I}v(2)$. For each $i\in I$ we have
$\Fun t(u(1),\Base i)\Exto{\Flate J}\Fun t(u(2),\Base i)$. Let
$w'\in\Extd{\Orthp{\Flate J}}$, we have by
Lemma~\ref{lemma:pcoh-extensional-let}
\begin{align*}
  \Eval{\Fun s(u(1),v(1))}{w'}
  = \Eval{v(1)}{v'(t,u(1),w')}
  \leq \Eval{v(2)}{v'(t,u(2),w')}
  = \Eval{\Fun s(u(2),v(2))}{w'}
\end{align*}
which proves our contention.
% where $x'\in\Intercc 01^{I\cup\Eset\Nerr}$ is defined by
% $x'_i=\Eval{\Fun t(u(2),\Base i)}{z'}$ for $i\in I$ and
% $x'_\Nerr=z'_\Nerr$. We have
% \begin{align*}
%   x'_i
%   &=\sum_{j\in J}\Fun t(u(2),\Base i)_jz'_j
%     +\Fun t(u(2),\Base i)_\Nerr z'_\Nerr\\
%   &\leq \sum_{j\in J}\Fun t(u(2),\Base i)_jz'_\Nerr
%     +\Fun t(u(2),\Base i)_\Nerr z'_\Nerr\leq z'_\Nerr=x'_\Nerr
% \end{align*}
% hence $x'\in\Extd{\Orth{\Flate I}}$ from which it follows that
% \begin{align*}
%   \Eval{\Fun s(u(1),x(1))}{z'}\leq\Eval{x(1)}{x'}\leq\Eval{x(2)}{x'}
%   =\Eval{\Fun s(u(2),x(2))}{z'}
% \end{align*}
% and hence $\Fun s(u(1),x(1))\Exto{\Flate J}\Fun s(u(2),x(2))$ as expected.

\begin{lemma}\label{lemma:let-pcoh-extensional-monotone}
  Let
  $s(1),s(2)\in\EPCOH(\Tens{\Excl
    X}{\Excl{\Flate I}},\Flate J)$. If
  $s(1)\Exto{\Tens{\Excl X}{\Excl{\Flate I}},\Flate J}s(2)$, then
  \[t(1)=\Slete{s(1)}\Exto{\Tens{\Excl
      X}{\Flate I},\Flate J}t(2)=\Slete{s(2)}\,.\]
\end{lemma}
\begin{proof}
  Let $u\in\Extd X$, $v\in\Extd{\Flate I}=\Pcohp{\Extc{\Flate I}}$ and
  $w'\in\Extd{\Orthp{\Flate J}}$. We have by
  Lemma~\ref{lemma:pcoh-extensional-let}
  \begin{align*}
    \Eval{\Fun{t(1)}(u,v)}{w'}
    = \Eval{v}{v'(s(1),u,w')} \leq\Eval{v}{x'(s(2),u,w')}
    = \Eval{\Fun{t(2)}(u,v)}{w'}\,.
  \end{align*}
\end{proof}

\subsubsection{Basic functions with error}
\label{sec:pcoh-ext-basic-functions}

With the same notations as above, let $f:I\to J$ be a partial function
of domain $D\subseteq I$. Then we define
$\Bfune f\in\Realpto{\Web{\Limpl{\Flate I}{\Flate J}}}$ as follows:
\begin{align*}
  \Bfune f_{a,b}=
  \begin{cases}
    1 & \text{if } a\in D\text{ and }b=f(a)\in J\\
    1 & \text{if } a=b=\Nerr\\
    0 & \text{otherwise.}
  \end{cases}
\end{align*}
We have $\Bfune f\in\EPCOH(\Flate I,\Flate J)$. Indeed let first
$u\in{\Extd{\Flate I}}=\Pcoh{\Extc{\Flate I}}$, we have
\begin{align*}
  \sum_{b\in\Web{\Extd{\Flate J}}}(\Matappa{\Bfune f}u)_b
  = \sum_{j\in J}(\Matappa{\Bfune f}u)_j+(\Matappa{\Bfune f}u)_\Nerr
  = \sum_{j\in J}\sum_{\Biind{i\in D}{f(i)=j}}u_j+x_\Nerr
  = \sum_{i\in D}u_i+u_\Nerr\leq 1
\end{align*}
and hence $\Matappa{\Bfune f}u\in\Pcoh{\Extd{\Flate J}}$. Assume now
that $v(1)\Exto{\Flate I}v(2)$ and let $w'\in\Extd{\Orthp{\Flate J}}$,
we have
\begin{align*}
  \Eval{\Matappa{\Bfune f}{v(1)}}{w'}
  = \sum_{j\in J}(\Matappa{\Bfune f}{x(1)})_jw'_j+u_\Nerr w'_\Nerr
  = \Eval{v(1)}{v'}
  % \sum_{i\in I}x_i\left(\sum_{\Biind{j\in J}{f(i)=j}}y'_j\right)
  %   +x_\Nerr y'_\Nerr\\
\end{align*}
where $v'\in\Realpto{\Web{\Extc{\Flate I}}}$ is defined by
$v'_i=w'_{f(i)}$ if $i\in D$, $v'_i=0$ if $i\in I\setminus D$ and
$v'_\Nerr=w'_\Nerr$ so that obviously $v'\in\Extd{\Orthp{\Flate
    I}}$. Therefore
$\Eval{\Matappa{\Bfune f}{v(1)}}{w'}\leq\Eval{\Matappa{\Bfune
    f}{v(2)}}{w'}$.

\subsubsection{Upper and lower approximating the identity at ground types.}
\label{sec:pcoh-ext-id-approx}
Let $J\subseteq I$, we define
$\Idl JI,\Idu JI\in\Realpto{\Web{\Extc{\Limpl{\Flate I}{\Flate I}}}}$
as follows:
\begin{align*}
  (\Idl JI)_{a,b}=
  \begin{cases}
    1 & \text{if }a=b\in J\cup\Eset{\Nerr}\\
    0 & \text{otherwise}
  \end{cases}
  \text{\quad and \quad}
  (\Idu JI)_{a,b}=
  \begin{cases}
    1 & \text{if }a=b\in J\cup\Eset{\Nerr}\\
    1 & \text{if }a\in I\setminus J\text{ and }b=\Nerr\\
    0 & \text{otherwise.}
  \end{cases}
\end{align*}
It is clear that
$\Idl JI,\Idu JI\in\Pcohp{\Limpl{\Extc{\Flate I}}{\Extc{\Flate I}}}$.
Notice that, if $v\in\Pcohp{\Extc{\Flate I}}$, we have
\begin{align*}
  \Fun{\Idl JI}(v)=\sum_{j\in J}v_j\Base j+v_\Nerr\Base\Nerr
  \text{\quad and \quad}
  \Fun{\Idu JI}(v)=\sum_{j\in J}v_j\Base j+
                    \left(\sum_{i\in I\setminus J}v_i+x_\Nerr\right)\Base\Nerr
\end{align*}
The fact that $\Idl JI\in\EPCOH(\Flate I,\Flate I)$ is a consequence
of Section~\ref{sec:pcoh-ext-basic-functions} applied to the
restriction of the identity function to $I$. Next we obviously have
$\Idl JI\leq\Id$ and hence $\Idl JI\Exto{\Flate I,\Flate I}\Id$.

\begin{lemma}\label{lemma:upper-id}
  One has $\Idu JI\in\EPCOH(\Flate I,\Flate J)$ and
  $\Id\Exto{\Flate I,\Flate I}\Idu JI$.
\end{lemma}

\section{Probabilistic PCF with errors}\label{sec:PCF-syntax}

For the sake of simplicity and readability our language of interest is
a simple and nonetheless very expressive probabilistic extension of
the well known Scott-Plotkin PCF language, extended with one
uncatchable exception $\Errconv$ of ground type called
\emph{convergence}. For convenience we also add a constant $\Errdiv$
for representing \emph{divergence} (it could be defined using the
$\Fix M$ construct). As in~\cite{EhrhardPaganiTasson18} the language
has a $\Let xMN$ construct allowing to deal with the ground data type
$\Tnat$ in a call-by-value manner; this is essential to implement
meaningful probabilistic algorithms in this language which is
call-by-name globally.
\begin{align*}
  \sigma,\tau,\dots & \Bnfeq \Tnat \Bnfor \Timpl\sigma\tau\\
  M,N,P,\dots
                    & \Bnfeq x \Bnfor \Num n \Bnfor \Succ M \Bnfor \Pred M
                      \Bnfor \If MNP \Bnfor \Let xMN\\
                    & \quad \Bnfor \App MN \Bnfor \Abst x\sigma M
                      \Bnfor \Fix M \Bnfor \Dice r
                      \Bnfor \Errdiv \Bnfor \Errconv
\end{align*}
The typing rules are given in
Figure~\ref{fig:PCF-typing-rules}. Notice the strong typing
restrictions on the conditional and $\mathsf{let}$ constructs, due to
the fact that our convergence and divergence constants are of ground
types. Given a typing context $\Gamma$ and a type $\sigma$, we use
$\Termsty\Gamma\sigma$ for the set of terms $M$ such that
$\Tseq\Gamma M\sigma$.

\begin{remark}
  This is not a restriction in term of expressiveness since more
  general versions of these constructs can be defined using
  $\lambda$-abstractions. Nevertheless this is clearly not
  satisfactory, and the solution is to develop a monadic description
  of the error $\Errconv$ in $\EPCOH$. This is not completely
  straightforward because this description has to be compatible with
  the extensional structures.
\end{remark}
\begin{figure}{\scriptsize
  \centering
  \AxiomC{$n\in\Nat$}
  \UnaryInfC{$\Tseq\Gamma{\Num n}\Tnat$}
  \DisplayProof
  \quad
  \AxiomC{$E\in\Eset{\Errdiv,\Errconv}$}
  \UnaryInfC{$\Tseq\Gamma{E}\Tnat$}
  \DisplayProof
  \quad
  \AxiomC{}
  \UnaryInfC{$\Tseq{x_1:\sigma_1,\dots,x_k:\sigma_k}{x_i}{\sigma_i}$}
  \DisplayProof
  \quad
  \AxiomC{$r\in\Intercc 01\cap\Rational$}
  \UnaryInfC{$\Tseq\Gamma{\Dice r}\Tnat$}
  \DisplayProof
  \quad
  \AxiomC{$\Tseq\Gamma M\Tnat$}
  \UnaryInfC{$\Tseq\Gamma{\Succ M}\Tnat$}
  \DisplayProof
  \Figbreak
  \AxiomC{$\Tseq\Gamma M\Tnat$}
  \UnaryInfC{$\Tseq\Gamma{\Pred M}\Tnat$}
  \DisplayProof
  \quad
  \AxiomC{$\Tseq\Gamma M\Tnat$}
  \AxiomC{$\Tseq\Gamma P\Tnat$}
  \AxiomC{$\Tseq\Gamma Q\Tnat$}
  \TrinaryInfC{$\Tseq\Gamma{\If MPQ}\Tnat$}
  \DisplayProof
  \quad
  \AxiomC{$\Tseq\Gamma M\Tnat$}
  \AxiomC{$\Tseq{\Gamma,x:\Tnat}P\Tnat$}
  \BinaryInfC{$\Tseq{\Gamma}{\Let xMP}\Tnat$}
  \DisplayProof
  \Figbreak
  \AxiomC{$\Tseq{\Gamma,x:\sigma}{M}{\tau}$}
  \UnaryInfC{$\Tseq\Gamma{\Abst x\sigma M}{\Timpl\sigma\tau}$}
  \DisplayProof
  \quad
  \AxiomC{$\Tseq\Gamma M{\Timpl\sigma\tau}$}
  \AxiomC{$\Tseq\Gamma N\sigma$}
  \BinaryInfC{$\Tseq\Gamma{\App MN}\tau$}
  \DisplayProof
  \quad
  \AxiomC{$\Tseq\Gamma M{\Timpl\sigma\sigma}$}
  \UnaryInfC{$\Tseq\Gamma{\Fix M}\sigma$}
  \DisplayProof}
  \caption{Typing rule for our PCF dialect}
  \label{fig:PCF-typing-rules}
\end{figure}

\subsection{Operational semantics}
% NB: pourquoi garder cette SO par réécriture alors qu'on a une MK
% plus loin?  Parce que cette SO parle aussi des points fixes alors
% que la MK est restreinte aux termes sans points fixes.
%
We equip this language with an operational semantics which is a
probabilistic extension of the usual deterministic weak head
reduction. More precisely, $M\Rel\Redwhd M'$ means that $M$ reduces to
$M'$ deterministically and $M\Rel{\Redwhp p}M'$ means that $M$ reduces
to $M'$ with probability $p\in\Intercc 01$. The reduction rules are
given in Figure~\ref{fig:PCF-reduction-rules}.
\begin{lemma}\label{lemma:redwhp-uniqueness}
  If $M\Rel{\Redwhp p}M'$ then $M$ and $M'$ are distinct terms, and
  $p$ can be recovered from $M$ and $M'$. 
\end{lemma}
This results from a simple inspection of the rule.
\begin{figure}{\scriptsize
  \centering
  \AxiomC{}
  \UnaryInfC{$\Succ{\Num n}\Rel\Redwhd\Num{n+1}$}
  \DisplayProof
  \quad
  \AxiomC{}
  \UnaryInfC{$\Pred{\Num 0}\Rel\Redwhd\Num{0}$}
  \DisplayProof
  \quad
  \AxiomC{}
  \UnaryInfC{$\Pred{\Num{n+1}}\Rel\Redwhd\Num{n}$}
  \DisplayProof
  \quad
  \AxiomC{$E\in\Eset{\Errdiv,\Errconv}$}
  \UnaryInfC{$\Succ E\Rel\Redwhd E$}
  \DisplayProof
  \quad
  \AxiomC{$E\in\Eset{\Errdiv,\Errconv}$}
  \UnaryInfC{$\Pred E\Rel\Redwhd E$}
  \DisplayProof
  \Figbreak
  \AxiomC{}
  \UnaryInfC{$\If{\Num 0}PQ\Rel\Redwhd P$}
  \DisplayProof
  \quad
  \AxiomC{}
  \UnaryInfC{$\If{\Num{n+1}}PQ\Rel\Redwhd Q$}
  \DisplayProof
  \quad
  \AxiomC{$E\in\Eset{\Errdiv,\Errconv}$}
  \UnaryInfC{$\If{E}PQ\Rel\Redwhd E$}
  \DisplayProof
  \quad
  \AxiomC{}
  \UnaryInfC{$\Let x{\Num n}M\Rel\Redwhd\Subst M{\Num n}x$}
  \DisplayProof
  \Figbreak
  \AxiomC{$E\in\Eset{\Errdiv,\Errconv}$}
  \UnaryInfC{$\Let x{E}M\Rel\Redwhd E$}
  \DisplayProof
  \quad
  \AxiomC{}
  \UnaryInfC{$\App{\Abst x\sigma M}N\Rel\Redwhd\Subst MNx$}
  \DisplayProof
  \quad
  \AxiomC{}
  \UnaryInfC{$\Fix M\Rel\Redwhd\App M{\Fix M}$}
  \DisplayProof
  \Figbreak
  \AxiomC{$M\Rel\Redwhd M'$}
  \UnaryInfC{$M\Rel{\Redwhp 1}M'$}
  \DisplayProof
  \quad
  \AxiomC{$r\in\Intercc 01\cap\Rational$}
  \UnaryInfC{$\Dice r\Rel{\Redwhp r}\Num 0$}
  \DisplayProof
  \quad
  \AxiomC{$r\in\Intercc 01\cap\Rational$}
  \UnaryInfC{$\Dice r\Rel{\Redwhp{1-r}}\Num 1$}
  \DisplayProof
  \quad
  \AxiomC{$M\Rel{\Redwhp p}M'$}
  \UnaryInfC{$\Succ M\Rel{\Redwhp p}\Succ{M'}$}
  \DisplayProof
  \quad
  \AxiomC{$M\Rel{\Redwhp p}M'$}
  \UnaryInfC{$\Pred M\Rel{\Redwhp p}\Pred{M'}$}
  \DisplayProof
  \Figbreak
  \AxiomC{$M\Rel{\Redwhp p}M'$}
  \UnaryInfC{$\If MPQ\Rel{\Redwhp p}\If{M'}PQ$}
  \DisplayProof
  \quad
  \AxiomC{$M\Rel{\Redwhp p}M'$}
  \UnaryInfC{$\Let xMP\Rel{\Redwhp p}\Let x{M'}P$}
  \DisplayProof}
  \caption{Weak head reduction for our PCF dialect}
  \label{fig:PCF-reduction-rules}
\end{figure}
A convergence path is a sequence $\gamma=(M_0,\dots,M_n)$ with
$n\in\Nat$ such that there are $p_0,\dots,p_{n-1}\in\Intercc 01$ such
that $\forall i\in\Eset{0,\dots,n-1}\ M_i\Rel{\Redwhp{p_i}}M_{i+1}$
and $M_n$ is $\Redwhp p$-normal for all $p$ (we simply say that $M_n$
is \Whnormal). By Lemma~\ref{lemma:redwhp-uniqueness}, the sequence
$(\List p0{n-1})$ is determined by $\gamma$. We use the following
notations: $\Cpathl\gamma=n$ (length), $\Cpaths\gamma=M_0$ and
$\Cpatht\gamma=M_n$ (source and target), and
$\Cpathp\gamma=\prod_{i=0}^{n-1}p_i$ (probability). Given
$M,P\in\Termsty\Gamma\sigma$ with $P$ \Whnormal, we use $\Cpathset MP$
for the set of all convergence paths $\gamma$ such that
$\Cpaths\gamma=M$ and $\Cpatht\gamma=P$. Then
\begin{align*}
  \Redwhpr MP=\sum_{\gamma\in\Cpathset MP}\Cpathp\gamma
\end{align*}
is the probability that $M$ reduces to
$P$. See~\cite{EhrhardPaganiTasson18} for a discrete Markov chain
interpretation of this definition\footnote{For a version of PCF without
  the exceptions $\Errconv$ and $\Errdiv$, but their addition to the
  language does not change the proofs and results.}, showing in
particular that this number is actually a probability.

\begin{example}
  Consider the term $M=\If{\Dice{\frac 12}}{\Num 0}{\Num 0}$. Then we
  have two \emph{distinct} convergence paths from $M$, with the same
  target $\Num 0$, namely $(M,\If{\Num 0}{\Num 0}{\Num 0},\Num 0)$ and
  $(M,\If{\Num 1}{\Num 0}{\Num 0},\Num 0)$. Both have probability
  $\frac 12$ and we have $\Redwhpr M{\Num 0}=1$.
\end{example}

A term $P\in\Termsty{}\Tnat$ is \Whnormal{} iff $P=\Num n$ for some
$n\in\Nat$ or $P\in\Eset{\Errdiv,\Errconv}$. For
$M\in\Termsty{}\Tnat$, we are mainly interested in evaluating
$\Redwhpr M\Errconv$, that we also denote as $\Redwhprc P$.

\subsection{A preorder relation on terms}
We define a binary relation $\Termso$ by the deduction rules of
Figure~\ref{fig:PCF-extensional-preorder}.
\begin{figure}{\footnotesize
  \centering
  \AxiomC{$M\in\Termsty\Gamma\Tnat$}
  \UnaryInfC{$\Errdiv\Termso M$}
  \DisplayProof
  \quad
  \centering
  \AxiomC{$M\in\Termsty\Gamma\Tnat$}
  \UnaryInfC{$M\Termso\Errconv$}
  \DisplayProof
  \quad
  \AxiomC{}
  \UnaryInfC{$x\Termso x$}
  \DisplayProof
  \quad
  \AxiomC{}
  \UnaryInfC{$\Num n\Termso \Num n$}
  \DisplayProof
  \quad
  \AxiomC{$r\in\Intercc 01\cap\Rational$}
  \UnaryInfC{$\Dice r\Termso\Dice r$}
  \DisplayProof
  \quad
  \AxiomC{$M\Termso M'$}
  \UnaryInfC{$\Succ M\Termso\Succ{M'}$}
  \DisplayProof
  \Figbreak
  \AxiomC{$M\Termso M'$}
  \UnaryInfC{$\Pred M\Termso\Pred{M'}$}
  \DisplayProof
  \quad
  \AxiomC{$M\Termso M'$}
  \AxiomC{$P\Termso P'$}
  \AxiomC{$Q\Termso Q'$}
  \TrinaryInfC{$\If MPQ\Termso\If{M'}{P'}{Q'}$}
  \DisplayProof
  \quad
  \AxiomC{$M\Termso M'$}
  \AxiomC{$P\Termso P'$}
  \BinaryInfC{$\Let xMP\Termso\Let x{M'}{P'}$}
  \DisplayProof
  \Figbreak
  \AxiomC{$M\Termso M'$}
  \AxiomC{$N\Termso N'$}
  \BinaryInfC{$\App MN\Termso\App{M'}{N'}$}
  \DisplayProof
  \quad
  \AxiomC{$M\Termso M'$}
  \UnaryInfC{$\Abst x\sigma M\Termso\Abst x\sigma{M'}$}
  \DisplayProof
  \quad
  \AxiomC{$M\Termso M'$}
  \UnaryInfC{$\Fix M\Termso\Fix{M'}$}
  \DisplayProof
  \Figbreak
  \AxiomC{$M\Termso M'$}
  \AxiomC{$\Fix M\Termso N'$}
  \BinaryInfC{$\Fix M\Termso\App{M'}{N'}$}
  \DisplayProof
  \quad
  \AxiomC{$M'\Termso M$}
  \AxiomC{$N'\Termso \Fix M$}
  \BinaryInfC{$\App{M'}{N'}\Termso\Fix M$}
  \DisplayProof\\}
  \caption{Extensional preorder on terms}
  \label{fig:PCF-extensional-preorder}
\end{figure}

\begin{lemma}
  If $\Tseq\Gamma M\sigma$ and $M'\Termso M$ or $M\Termso M'$, then
  $\Tseq\Gamma{M'}\sigma$.
\end{lemma}
\begin{proof}
  Easy induction on the derivation of $M'\Termso M$ or $M\Termso M'$.
\end{proof}

The following property is natural and worth being noticed, though it
plays no technical role in the sequel.
\begin{proposition}
  The relation $\Termso$ is a preorder relation on
  $\Termsty\Gamma\sigma$, for each context $\Gamma$ and type $\sigma$.
\end{proposition}
% \begin{proof}[Sketch of proof]
%   Reflexivity is proven by a simple induction on terms (or rather, on
%   typing derivations). For transitivity one assumes that
%   $M\Termso N\Termso P$ for $M,N,P\in\Termsty\Gamma\sigma$ and one
%   reasons by induction on the maximal height of the two derivations in
%   the deduction system of
%   Figure~\ref{fig:PCF-extensional-preorder}. The only non-trivial
%   cases are when $M=\Fix{M_0}$ and when $M=\App{M_0}{M_1}$.
% \end{proof}

\subsection{Denotational semantics in
  $\EPCOH$}\label{sec:extensional-pcoh-semantics}
We use the functions $\Fsucc,\Fpred:\Nat\to\Nat$:
$\Fsucc(n)=n+1$, $\Fpred(0)=0$ and $\Fpred(n+1)=n$.

Given a type $\sigma$, we define an extensional PCS $\Tsem\sigma$ by
induction: $\Tsem\Tnat=\Flate\Nat$ and
$\Tsem{\Timpl\sigma\tau}=\Simpl{\Tsem\sigma}{\Tsem\tau}$.  Given a
context $\Gamma=(x_1:\sigma_1,\dots,x_k:\sigma_k)$, we define
$\Tsem\Gamma=\Bwith_{i=1}^k\Tsem{\sigma_i}$ which is an object of
$\EPCOH$.

Next given $M\in\Termsty\Gamma\tau$ we define
$\Psem M\Gamma\in\Kl\EPCOH(\Tsem\Gamma,\Tsem\tau)$ by induction on
$M$. We know that this morphism is fully characterized by the
associated function
$f=\Fun{\Psem M\Gamma}:\Pcohp{\Bwith_{i=1}^k\Extc{\Tsem{\sigma_i}}}
\to\Pcoh{\Extc{\Tsem\tau}}$. Let $C=\Bwith_{i=1}^k\Extc{\Tsem{\sigma_i}}$.

\Itmath
If $M=x_i$ then $\Psem M\Gamma$ is defined as the following
composition of morphisms in $\EPCOH$:\\
\begin{tikzcd}
  \Excl C\arrow[r,"\Der C"] & C\arrow[r,"\Proj i"] & \Tsem{\sigma_i} 
\end{tikzcd},
so that $f(\Vect u)=u_i$.

\Itmath
If $M=\Num n$ then $\Psem M\Gamma$ is defined as the following
composition of morphisms in $\EPCOH$:\\
\begin{tikzcd}
  \Excl C\arrow[r,"\Weak C"] & \One\arrow[r,"\Base n"] & \Flate\Nat
\end{tikzcd} so that $f(\Vect u)=\Base n$.

\Itmath
If $M=\Errdiv$ then $\Psem M\Gamma=0$ so that $f(\Vect u)=0$.

\Itmath
If $M=\Errconv$ then $\Psem M\Gamma$ is defined as the following
composition of morphisms in $\EPCOH$:\\
\begin{tikzcd}
  \Excl C\arrow[r,"\Weak C"] & \One\arrow[r,"\Base \Nerr"] & \Flate\Nat
\end{tikzcd} so that $f(\Vect u)=\Base\Nerr$.

\Itmath
If $M=\Succ N$ with $N\in\Termsty\Gamma\Tnat$ then $\Psem M\Gamma$ is
defined as the following composition of morphisms in $\EPCOH$:
\begin{tikzcd}
  \Excl C\arrow[r,"\Psem N\Gamma"]
  & \Flate\Nat\arrow[r,"\Bfune\Fsucc"] & \Flate\Nat
\end{tikzcd} so that
$f(\Vect u)=\sum_{n\in\Nat}g(\Vect u)_n\Base{n+1}+g(\Vect
u)_\Nerr\Base\Nerr$ where $g=\Fun{\Psem N\Gamma}$.

\Itmath
If $M=\Pred N$ with $N\in\Termsty\Gamma\Tnat$ then $\Psem M\Gamma$ is
defined as the following composition of morphisms in $\EPCOH$:
\begin{tikzcd}
  \Excl C\arrow[r,"\Psem N\Gamma"] &
  \Flate\Nat\arrow[r,"\Bfune\Fpred"] & \Flate\Nat
\end{tikzcd} so that
$f(\Vect u)=g(\Vect u)_0\Base 0+\sum_{n\in\Nat}g(\Vect
u)_{n+1}\Base{n}+g(\Vect u)_\Nerr\Base\Nerr$ where
$g=\Fun{\Psem N\Gamma}$.

\Itmath
Assume that $M=\If N{P_1}{P_2}$ with
$N,P_1,P_2\in\Termsty\Gamma\Tnat$.  Let $s=\Psem N\Gamma$ and, for
$n\in\Nat$, let $t_n\in\Kl\EPCOH(C,\Flate\Nat)$ be defined by
$t_0=\Psem{P_1}\Gamma$ and $t_{n+1}=\Psem{P_2}\Gamma$ for each
$n\in\Nat$. Let
$t=\Tuple{t_n}_{n\in\Nat}\in\Kl\EPCOH(C,(\Flate\Nat)^\Nat)$. Then
$\Psem M\Gamma$ is defined as the following composition of morphisms
in $\EPCOH$:\\
\begin{tikzcd}
  \Excl C\arrow[r,"\Contr C"]
  & \Tens{\Excl C}{\Excl C}\arrow[r,"\Tens{s}{\Prom t}"]
  & \Tens{\Flate\Nat}{\Exclp{(\Flate\Nat)^\Nat}}\arrow[r,"\Pcasee{\Nat,\Nat}"]
  & \Flate\Nat
\end{tikzcd}\\
so that
$f(\Vect u)=g(\Vect u)_0h_1(\Vect u)+\left(\sum_{n=0}^\infty g(\Vect
  u)_{n+1}\right)h_2(\Vect u)+g(\Vect u)_\Nerr\Base\Nerr$ where
$g=\Fun{\Psem N\Gamma}$ and $h_i=\Fun{\Psem{P_i}\Gamma}$ for $i=1,2$.

\Itmath
Assume that $M=\Let xNP$ with $N\in\Termsty\Gamma\Tnat$ and
$P\in\Termsty{\Gamma,x:\Tnat}\Tnat$. Let
$s=\Psem N\Gamma\in\Kl\EPCOH(C,\Flate\Nat)$ and
$t=\Psem P{\Gamma,x:\Tnat}\in\Kl\EPCOH(\With
C{\Flate\Nat},\Flate\Nat)$ so that
$t\Compl\Seelyt\in\EPCOH(\Tens{\Excl
  C}{\Excl{\Flate\Nat}},\Flate\Nat)$. Then $\Psem M\Gamma$ is defined
as the following composition of morphisms in $\EPCOH$:\\
\begin{tikzcd}
    \Excl C\arrow[r,"\Contr C"] & \Tens{\Excl C}{\Excl
      C}\arrow[r,"\Tens{\Excl C}{s}"] & \Tens{\Excl
      C}{\Flate\Nat}\arrow[r,"\Slete{t\Compl\Seelyt}"] & \Flate\Nat
\end{tikzcd}\\
so that
$f(\Vect u)=\sum_{n\in\Nat}g(\Vect u)_nh(\Vect u,\Base n)+g(\Vect
u)_\Nerr\Base\Nerr$ where $g=\Fun{\Psem N\Gamma}$ and
$h=\Fun{\Psem P{\Gamma,x:\Tnat}}$.

\Itmath
Assume that $M=\App NP$ with $N\in\Termsty\Gamma{\Timpl\sigma\tau}$
and $P\in\Termsty\Gamma\sigma$.  Let
$s=\Psem N\Gamma\in\Kl\EPCOH(C,\Tsem{\Timpl\sigma\tau})$ and
$t=\Psem P{\Gamma}\in\Kl\EPCOH(C,\Tsem\sigma)$.  Then $\Psem M\Gamma$
is defined as the following composition of morphisms in $\EPCOH$:\\
\begin{tikzcd}
  \Excl C\arrow[r,"\Contr C"]
  & \Tens{\Excl C}{\Excl C}\arrow[r,"\Tens{s}{\Prom t}"]
  & \Tens{(\Limpl{\Excl{\Tsem\sigma}}{\Tsem\tau})}
  {\Excl{\Tsem\sigma}}\arrow[r,"\Evlin"]
  & \Tsem\tau
\end{tikzcd}\\
so that $f(\Vect u)=\Fun{g(\Vect u)}(h(\Vect u))$ where
$g=\Fun{\Psem N\Gamma}$ and $h=\Fun{\Psem t\Gamma}$.

\Itmath
If $M=\Abst x\sigma N$ with $N\in\Termsty{\Gamma,x:\sigma}{\phi}$ (so
that $\tau=(\Timpl\sigma\phi)$), let
$s=\Psem{N}{\Gamma,x:\sigma}\in\Kl\EPCOH(\With
C{\Tsem\sigma},\Tsem\phi)$, so that
$s\Compl\Seelyt\in\EPCOH(\Tens{\Excl
  C}{\Excl{\Tsem\sigma}},\Tsem\phi)$.  Then
$\Psem M\Gamma=\Curlin(s\Compl\Seelyt)\in\Kl\EPCOH(C,\Tsem\tau)$ so
that $f(\Vect u)$ is the element of $\Extd{\Tsem{\Timpl\sigma\phi}}$
characterized by $\forall u\in\Pcoh{\Extc{\Tsem{\sigma}}}$
$\Fun{f(\Vect u)}=g(\Vect u,u)$ where
$g=\Fun{\Psem N{\Gamma,x:\sigma}}$.

\Itmath
If $M=\Fix N$ with $N\in\Termsty\Gamma{\Timpl\tau\tau}$, let
$s=\Psem M\Gamma\in\Kl\EPCOH(C,\Simpl{\Tsem\tau}{\Tsem\tau})$. Then
$\Psem M\Gamma=\Fixpcoh{\Extc{\Tsem\tau}}\Compl s$ which belongs to
$\Kl\EPCOH(C,\Tsem\tau)$ by Theorem~\ref{th:ext-pcoh-fixpoint} so that
$f(\Vect u)=\sup_{n\in\Nat}h^n(0)$ where $h=\Fun{g(\Vect u)}$, where
$g=\Fun{\Psem N\Gamma}$. Notice that $t=\Psem M\Gamma$ satisfies the
following commutation in $\EPCOH$:
\begin{equation}\label{eq:fixpoint-diag}
\begin{tikzcd}
  \Excl C\arrow[r,"\Contr C"]\arrow[rrd,swap,"t"]
  & \Tens{\Excl C}{\Excl C}\arrow[r,"\Tens s{\Prom t}"]
  & \Tens{\Limplp{\Excl{\Tsem\tau}}{\Tsem\tau}}{\Excl{\Tsem\tau}}
  \arrow[d,"\Evlin"]\\
  && \Tsem\tau
\end{tikzcd}
\end{equation}

\begin{lemma}[Substitution]\label{lemma:PCF-substitution}
  If $\Tseq{\Gamma,x:\sigma}M\tau$ and $\Tseq\Gamma N\sigma$ then
  $\Psem{\Subst MNx}\Gamma$ coincides with the composition
  of morphisms (setting $C=\Tsem\Gamma$)
  \begin{tikzcd}
    \Excl C\ar[r,"\Contr C"]
    & \Excl C\ITens\Excl C\ar[r,"\Id\ITens\Prom{\Psem N\Gamma}"]
    &[1em] \Excl C\ITens\Excl{\Tsem\sigma}\ar[r,"\Psem M\Gamma"]
    & \Tsem\tau
  \end{tikzcd}
  where we keep implicit the Seely isomorphism.
\end{lemma}
\begin{proof}
  Simple induction on $M$.
\end{proof}

\begin{theorem}
  If $M\in\Termsty{}\Tnat$ then ${\Psem M{}}_\Nerr=\Redwhprc M$.
\end{theorem}
\begin{proof}
  Follows exactly the same pattern as the proof of adequacy
  in~\cite{EhrhardPaganiTasson18}.
\end{proof}

\begin{theorem}\label{th:psem-ext-order}
  Let $\Gamma=(x_1:\sigma_1,\dots,x_k:\sigma_k)$ be a typing context
  and assume that $M_1,M_2\in\Termsty{}\tau$ and $M_1\Termso
  M_2$. Then for all $\Vect u\in\prod_{i=1}^k\Extd{\Tsem{\sigma_i}}$
  one has
  $\Fun{\Psem{M_1}\Gamma}(\Vect
  u)\Exto{\Tsem\tau}\Fun{\Psem{M_2}\Gamma}(\Vect u)$.
\end{theorem}
\begin{proof}
  By induction on the height of the proof of $M\Termso N$ in the
  deduction system of Figure~\ref{fig:PCF-extensional-preorder}.

  If $M=\Errdiv$ and $N$ is any term such that $\Tseq\Gamma N\Tnat$
  then we use the fact that $\emptyset$ is
  $\Exto{\Flate\Nat}$-minimal.
  
  Assume $M=x_i$, then $N=x_i$ or $N=\Errconv$ (if
  $\sigma_i=\Tnat$). In the first case we use reflexivity of
  $\Exto{\Tsem{\tau_i}}$ (see
  Proposition~\ref{prop:pre-ext-basic-props}) and in the first case we
  use the fact that $\Base\Nerr$ is $\Exto{\Flate\Nat}$-maximal. The
  cases $M=\Num n$ and $M=\Dice r$ are similar.

  Assume that $M=\Succ{M_0}$ so that $\tau=\Tnat$ and
  $\Tseq\Gamma{M_0}\Tnat$. Then we have either $N=\Errconv$ or
  $N=\Succ{N_0}$ with $M_0\Termso N_0$. In the first case we use as
  above the $\Exto{\Flate\Nat}$-maximality of $\Base\Nerr$. In the
  second case, we use the inductive hypothesis, the definition of
  $\Psem M\Gamma$ and the $\Exto{\Flate\Nat}$-monotonicity of
  $\Bfune\Fsucc$, see Section~\ref{sec:pcoh-ext-basic-functions}. The
  cases $M=\Pred{M_0}$ and $\If{M_0}{M_1}{M_2}$ are similar, using
  Section~\ref{sec:pcoh-ext-case} for the conditional.

  Assume now that $M=\Let x{M_0}{M_1}$ with $\Tseq\Gamma{M_0}\Tnat$
  and $\Tseq{\Gamma,x:\Tnat}{M_1}{\Tnat}$. The case $N=\Nerr$ is dealt
  with as above so assume that $N=\Let x{N_0}{N_1}$ with
  $M_i\Termso N_i$ for $i=0,1$. By inductive hypothesis we have
  $\Psem{M_1}{\Gamma,x:\Tnat}
  \Exto{\Excl{\Tsem{\Gamma,x:\Tnat},\Tsem\Tnat}}\Psem{N_1}{\Gamma,x:\Tnat}$
  and hence
  $\Psem{M_1}{\Gamma,x:\Tnat}\Compl\Seelyt
  \Exto{\Tens{\Excl{\Tsem{\Gamma,x:\Tnat}}}{\Excl{\Flate\Nat}},\Flate\Nat}
  \Psem{N_1}{\Gamma,x:\Tnat}\Compl\Seelyt$. It follows that
  $\Slete{\Psem{M_1}{\Gamma,x:\Tnat}\Compl\Seelyt}
  \Exto{\Tens{\Excl{\Tsem{\Gamma,x:\Tnat}}}{\Flate\Nat},\Flate\Nat}
  \Slete{\Psem{N_1}{\Gamma,x:\Tnat}\Compl\Seelyt}$ by
  Section~\ref{sec:pcoh-ext-let}. By inductive hypothesis we have also
  $\Psem{M_1}\Gamma\Exto{\Excl{\Tsem\Gamma},\Flate\Nat}\Psem{N_1}\Gamma$
  and hence
  $\Psem M\Gamma\Exto{\Excl{\Tsem\Gamma},\Flate\Nat}\Psem N\Gamma$ by
  definition of these interpretations and $\Exto{}$-monotonicity of
  composition. The case $M=\Abst x\sigma{M_0}$ with
  $\Tseq{\Gamma,x:\sigma}{M_0}\phi$ (and hence
  $\tau=\Timpl\sigma\phi$) is similar.

  Assume that $M=\App{M_0}{M_1}$ with
  $\Tseq\Gamma{M_0}{\Timpl\sigma\tau}$ and
  $\Tseq\Gamma{M_1}\sigma$. The case $N=\Errconv$ (and hence
  $\tau=\Tnat$) is dealt with as above. Assume that $N=\App{N_0}{N_1}$
  with $M_i\Termso N_i$ for $i=0,1$. By inductive hypothesis
  $\Psem{M_1}\Gamma\Exto{\Excl{\Tsem\Gamma},\Tsem\sigma}\Psem{N_1}\Gamma$
  and hence
  $\Prom{\Psem{M_1}\Gamma}\Exto{\Excl{\Tsem\Gamma},\Excl{\Tsem\sigma}}
  \Prom{\Psem{N_1}\Gamma}$ (see
  Section~\ref{sec:pcoh-ext-enriched-model}). We also have
  $\Psem{M_0}\Gamma
  \Exto{\Excl{\Tsem\Gamma},\Limpl{\Excl{\Tsem\sigma}}{\Tsem\tau}}
  \Psem{N_0}\Gamma$ and hence
  $\Psem M\Gamma\Exto{\Excl{\Tsem\Gamma},\Tsem\tau}\Psem N\Gamma$ by
  definition of these interpretations and $\Exto{}$-monotonicity of
  $\Evlin$. Assume last that $N=\Fix{N_0}$ with $M_0\Termso N_0$ and
  $M_1\Termso N$ (so that $\sigma=\tau$ and
  $\Tseq\Gamma{N_0}{\Timpl\tau\tau}$), by derivations shorter than
  that of $M\Termso N$. Setting $C=\Tsem\Gamma$, $s=\Psem{N_0}\Gamma$
  and $t=\Psem{N}\Gamma$, we know that
  Diagram~\Eqref{eq:fixpoint-diag} commutes. By inductive hypothesis,
  we have
  $\Psem{M_0}\Gamma\Exto{\Excl
    C,\Limpl{\Excl{\Tsem\tau}}{\Tsem\tau}}\Psem{N_0}\Gamma$ and
  $\Psem{M_1}\Gamma\Exto{\Excl C,\Tsem\tau}\Psem N\Gamma$ and hence
  \begin{align*}
    \Psem M\Gamma
    = \Evlin\Compl\Tensp{\Psem{M_0}\Gamma}{\Prom{\Psem{M_1}\Gamma}}
      \Compl\Contr C
    \Exto{\Excl C,\Tsem\tau}
      \Evlin\Compl\Tensp{\Psem{N_0}\Gamma}{\Prom{\Psem{N}\Gamma}}
      \Compl\Contr C=\Psem N\Gamma\,.
  \end{align*}

  Assume last that $M=\Fix{M_0}$ with
  $\Tseq\Gamma{M_0}{\Timpl\tau\tau}$. The case $N=\Errconv$ (and hence
  $\tau=\Tnat$) is dealt with as above. Assume $N=\Fix{N_0}$ with
  $M_0\Termso N_0$. By inductive hypothesis we have
  $\Psem{M_0}\Gamma\Exto{\Excl C,\Limpl{\Excl{\Tsem\tau}}{\Tsem\tau}}
  \Psem{N_0}\Gamma$ and hence
  \begin{align*}
    \Psem M\Gamma=\Fixpcoh{\Tsem\tau}\Compl\Prom{\Psem{M_0}\Gamma}
    \Exto{\Excl C,\Tsem\tau}
    \Fixpcoh{\Tsem\tau}\Compl\Prom{\Psem{N_0}\Gamma}=\Psem N\Gamma
  \end{align*}
  by Section~\ref{sec:pcoh-ext-enriched-model} and
  Theorem~\ref{th:ext-pcoh-fixpoint}. The last case is
  $N=\App{N_0}{N_1}$ with $M_0\Termso N_0$, $M_0\Termso N_0$ and
  $M\Termso N_1$, which is dealt with as the case $M=\App{M_0}{M_1}$
  and $N=\Fix{N_0}$ above.
\end{proof}

Combining the two previous theorems we get:
\begin{theorem}
  If $M,N\in\Termsty{}\Tnat$ satisfy $M\Termso N$, one has
  $\Redwhprc M={\Psem M{}}_\Nerr\leq{\Psem N{}}_\Nerr=\Redwhprc
  N$.
\end{theorem}

\section{Approximating probabilities of convergence with a Krivine machine}
\label{sec:Krivine-machine}

\subsection{A Krivine machine computing polynomials}
Here we present the outputs of the machine which are trees
representing some kind of polynomials, and its inputs which, not
surprisingly are pairs made of a term and a stack. The main
peculiarities to keep in mind is that these states are not closed but
only ``almost closed'' in the sense that all free variables have ground
type.

\subsubsection{Infinite polynomials.}
Let $I$ be an at most countable set and $V$ be a finite set of
variables.  Intuitively, a variable $x\in V$ represents an $I$-indexed
family of scalars taken in some fixed semiring $\Bb K$ (in the sequel,
$\Bb K$ will be $\Realp$). So for each $x\in V$ and $i\in I$ we
introduce the notation $x(i)$ representing intuitively the $i$th
component of $x$. A multi-exponent is a family
$\Vect\mu=(\mu_x)_{x\in V}\in\Mfin I^V$. The support of $\Vect\mu$ is
$\Supp{\Vect\mu}=\Union_{x\in V}\Eset{i\in I\St\mu_x(i)\not=0}$, which
is a finite subset of $I$. The associated monomial is the formal
commutative product
$V^{\Vect\mu}=\prod_{x\in V}\prod_{i\in I}x(i)^{\mu_x(i)}$. As an
example, take $I=\Nat$ and $V=\Eset{x,y,z}$. A typical multi-exponent
is $\Vect\mu=(\mu_x,\mu_y,\mu_z)$ such that $\mu_x=\Mset{1,1,3}$,
$\mu_y=\Mset{2,4,4}$ and $z=\Mset{1,1,1,1}$. The associated monomial
is $V^{\Vect\mu}=x(1)^2x(3)y(2)y(4)^2z(1)^4$. Formally, a polynomial
is a family $(\alpha_{\Vect\mu})$ of scalars $\in\Bb K$, indexed by
monomials, such that for any finite subset $I_0$ of $I$, the set
$\Eset{\Vect\mu\St \alpha_{\Vect\mu}\not=0\text{ and
  }\Supp{\Vect\mu}\subseteq I_0}$ is finite; this condition will be
called the \emph{finiteness condition on polynomials} (this
corresponds to what was called finite depth in the Introduction).  In
standard algebraic notations, such a polynomial is written as the
formal sum $\sum_{\Vect\mu}\alpha_{\Vect\mu}V^{\Vect\mu}$.

Notice that if $I$ is finite, a polynomial is just a finite linear
combination of monomials, but this is no more true when $I$ is
infinite. Here is a typical example of (infinite) polynomial:
$\sum_{n\in\Nat}x(n)^n$. On the other hand, the expression
$\sum_{n\in\Nat}x(1)^n$ is not a polynomial.

We use $\Polynoms{\Bb K}VI$ for the set of these polynomials. Equipped
with usual addition and multiplication this set is a $\Bb
K$-algebra. Let indeed
$A=\sum_{\Vect\mu}\alpha_{\Vect\mu}V^{\Vect\mu}$ and
$B=\sum_{\Vect\mu}\beta_{\Vect\mu}V^{\Vect\mu}$ be polynomials. Then
$A+B=\sum_{\Vect\mu}(\alpha_{\Vect\mu}+\beta_{\Vect\mu})V^{\Vect\mu}$
is a polynomial because if $\alpha_{\Vect\mu}+\beta_{\Vect\mu}\not=0$
then $\alpha_{\Vect\mu}\not=0$ or $\beta_{\Vect\mu}\not=0$ and hence
$A+B$ satisfies the finiteness condition on polynomials. As to
products, observe first that, given a multi-exponent $\Vect\mu$, that
there are only finitely many pairs of multi-exponents
$(\Vect\lambda,\Vect\rho)$ such that
$\Vect\lambda+\Vect\rho=\Vect\mu$. So we can set
$AB=\sum_{\Vect\mu}(\sum_{\Vect\lambda+\Vect\rho
  =\Vect\mu}\alpha_{\Vect\lambda}\beta_{\Vect\rho})V^\mu$. Let us
check that $AB$ satisfies the finiteness condition on polynomials, so
let $I_0$ be a finite subset of $I$. If $\Vect\mu$ is such that
$\sum_{\Vect\lambda+\Vect\rho=\Vect\mu}\alpha_{\Vect\lambda}\beta_{\Vect\rho}\not=0$
then there is a pair $(\Vect\lambda,\Vect\rho)$ such that
$\alpha_{\Vect\lambda}\not=0$ and $\beta_{\Vect\rho}\not=0$ and
$\Vect\lambda+\Vect\rho=\Vect\mu$. If $\Supp{\Vect\mu}\subseteq I_0$
we also have $\Supp\lambda,\Supp\rho\subseteq I_0$. Since $A$ and $B$
are polynomials, there are only finitely many such pairs
$(\Vect\lambda,\Vect\rho)$ and hence there are only finitely many
$\Vect\mu$ such that
$\sum_{\Vect\lambda+\Vect\rho=
  \Vect\mu}\alpha_{\Vect\lambda}\beta_{\Vect\rho}\not=0$ and
$\Supp{\Vect\mu}\subseteq I_0$.

Let $\Vect v$ be a valuation for variables $x\in V$ in $\Bb K^{(I)}$,
that is, for each $x\in V$, $v_x$ is an $I$-indexed family
$(v_x(i))_{i\in I}$ of elements of $\Bb K$ which vanishes almost
everywhere. We set
$v^{\Vect\mu}=\prod_{x\in V}\prod_{i\in I}v_x(i)^{\mu_x(i)}\in\Bb K$;
this is a finite product in $\Bb K$ (it would be the case even for
$v_x\in\Bb K^I$). Then by the finiteness property the sum
$A(\Vect v)=\sum_{\Vect\mu}\alpha_{\Vect\mu}v^{\Vect\mu}$ has only
finitely many non-vanishig terms and hence $A(\Vect v)\in\Bb K$ is
well-defined. This is the main motivation for the finiteness condition
in the definition of infinite polynomials.

\begin{lemma}\label{lemma:polynoms-variable-cases}
  If, for each $i\in I$, $A_i\in\Polynoms{\Bb K}VI$ and $x\in V$, then
  $A=\sum_{i\in I}x(i)A_i\in\Polynoms{\Bb K}VI$.
\end{lemma}
\begin{proof}
  If $i\in I$, we use $\Meone ix$ for the multiexponent $\Vect\mu$
  such that $\mu_x=\Mset i$ and $\mu_y=\Mset{}$ for $y\not=x$, in
  other words $V^{\Meone ix}=x(i)$.  Let us write
  $A_i=\sum_{\Vect\mu}\alpha(i)_{\Vect\mu}V^{\Vect\mu}$. We have
  \begin{align*}
    A
    = \sum_{i\in I}\sum_{\Vect\mu}\alpha(i)_{\Vect\mu}V^{\Vect\mu+\Meone ix}
    = \sum_{\Vect\nu}\left(\sum_{i\in\Supp{\nu_x}}
    \alpha(i)_{\Vect{\nu}-\Meone ix}\right)V^{\Vect\nu}\,.
  \end{align*}
  Since $\Supp{\nu_x}$ is a finite set, each coefficient
  $\beta_{\Vect\nu}=\sum_{i\in\Supp{\nu_x}} \alpha_{\Vect{\nu}-\Meone
    ix}$ in this expression is a finite sum. Let $I_0\subseteq I$ be
  finite.  Let $M$ be the set of all $\Vect\nu$'s such that
  $\Supp{\Vect\nu}\subseteq I_0$ and $\beta_{\Vect\nu}\not=0$. For
  each $\Vect\nu\in M$ there must be $i(\Vect\nu)\in I$ such that
  $i(\Vect\nu)\in\Supp{\nu_x}$ and
  $\alpha(i(\Vect\nu))_{\Vect\nu-\Meone {i(\Vect\nu)}x}\not=0$. Assume
  $M$ is infinite. We have $\forall\Vect\nu\in M\ i(\Vect\nu)\in I_0$
  and hence, since $I_0$ is finite, there must be $i\in I_0$ such that
  $i(\Vect\nu)=i$ for all $\Vect\nu\in M'$ where $M'$ is an infinite
  subset of $M$ (pigeonhole principle). Since the map
  $\Vect\nu\mapsto\Vect\nu-\Meone ix$ is injective on $M'$, since
  $\Supp{\Vect\nu-\Meone ix}\subseteq I_0$ and since
  $\alpha(i)_{\Vect\nu-\Meone ix}\not=0$ for all $\Vect\nu\in M'$, we
  have a contradiction with the finiteness condition on polynomials
  satisfied by $A_i$. Hence $M$ is finite and $A$ satisfies the
  finiteness condition on polynomials.
\end{proof}

\subsubsection{Well-founded trees and polynomials.} Our Krivine
machine will produce polynomials presented as some kind of
well-founded Böhm trees that we define now, we call
them polynomial trees. They are generated by the following syntax:
\begin{itemize}
\item $\Ptone$ and $\Ptzero$ are polynomial trees;
\item if $x\in V$ and $\cS$ is a function from $I$ to
  polynomial trees then $\Ptvar x\cS$ is a polynomial tree;
\item if $\alpha_1,\alpha_2\in\Bb K$ and $S_1,S_2$ are polynomial
  trees then $\Ptcl{\alpha_1}{S_1}{\alpha_2}{S_2}$ is a polynomial
  tree.
\end{itemize}
We use $\Polynomts {\Bb K}VI$ for the set of these trees.  We define a
map $\Polyofpt:\Polynomts{\Bb K}VI\to\Polynoms{\Bb K}VI$ by
well-founded induction: $\Polyofpt(\Ptone)=1$,
$\Polyofpt(\Ptcl{\alpha_1}{S_1}{\alpha_2}{S_2}=
\alpha_1\Polyofpt(S_1)+\alpha_2\Polyofpt(S_2)$ and
$\Polyofpt(\Ptvar x\cS)=\sum_{i\in I}x(i)\Polyofpt(\cS(i))$. The fact
that $\Polyofpt(S)\in\Polynoms{\Bb K}VI$ results from
Lemma~\ref{lemma:polynoms-variable-cases}.

\subsubsection{Stacks, states and their denotational semantics}
Our machine is restricted to almost free and \Fixfree{} PCF terms: a
term, a stack or a state is \Fixfree{} if it does not contain the
$\Fix\_$ construct and almost closed if its only variables are of type
$\Tnat$.  The machine computes polynomial trees
$\in\Polynomts{\Realp}{V}{\Natt}$ where $\Natt=\Nat\cup\Eset\Nerr$.

% A \emph{state} of the
% machine is a pair $\State M\pi$ where $M$ is a \Fixfree{} and almost
% closed term and $\pi$ is a \Fixfree{} and almost closed
% \emph{stack}.

However it is more natural to define stacks, states and their
semantics without these additional restrictions which will be useful
only for the machine itself. General stacks are given by the following
grammar
\begin{align*}
  \pi,\rho\dots
  \Bnfeq \Stempty
  \Bnfor \Stsucc\pi \Bnfor \Stpred\pi \Bnfor \Stif PQ\pi
  \Bnfor \Stlet xM\pi
  \Bnfor \Starg M\pi
\end{align*}
where $M$, $P$ and $Q$ are terms. Stacks and states are typed
by the rules of Figure~\ref{fig:PCFstacks-typing-rules}.

The judgment $\Tseqst\Gamma\pi\sigma$ means that $\pi$ is a
continuation which expects an argument of type $\sigma$ in context
$\Gamma$. The judgment $\Tseqstate\Gamma q$ means that the state $q$
is well typed in context $\Gamma$, its ``type'' is left implicit (it
would be the formula $\Llbot$ of linear logic). The empty stack
$\Stempty$ should be understood as a continuation which catches the
$\Errconv$ ``exception''.
\begin{figure}{\footnotesize
  \centering
  \AxiomC{}
  \UnaryInfC{$\Tseqst\Gamma\Stempty\Tnat$}
  \DisplayProof
  \quad
  \AxiomC{$\Tseqst\Gamma\pi\Tnat$}
  \UnaryInfC{$\Tseqst\Gamma{\Stsucc\pi}\Tnat$}
  \DisplayProof
  \quad
  \AxiomC{$\Tseqst\Gamma\pi\Tnat$}
  \UnaryInfC{$\Tseqst\Gamma{\Stpred\pi}\Tnat$}
  \DisplayProof
  \quad
  \AxiomC{$\Tseq {\Gamma}P\Tnat$}
  \AxiomC{$\Tseq {\Gamma}Q\Tnat$}
  \AxiomC{$\Tseqst \Gamma\pi\Tnat$}
  \TrinaryInfC{$\Tseqst\Gamma{\Stif PQ\pi}\Tnat$}
  \DisplayProof
  \Figbreak
  \AxiomC{$\Tseq{\Gamma,x:\Tnat}M\Tnat$}
  \AxiomC{$\Tseqst\Gamma\pi\Tnat$}
  \BinaryInfC{$\Tseqst\Gamma{\Stlet xM\pi}\Tnat$}
  \DisplayProof
  \quad
  \AxiomC{$\Tseq{\Gamma}M\sigma$}
  \AxiomC{$\Tseqst\Gamma\pi\tau$}
  \BinaryInfC{$\Tseqst\Gamma{\Starg M\pi}{\Timpl\sigma\tau}$}
  \DisplayProof  \quad
  \AxiomC{$\Tseq{\Gamma}M\sigma$}
  \AxiomC{$\Tseqst\Gamma\pi\sigma$}
  \BinaryInfC{$\Tseqstate\Gamma{\State M\pi}$}
  \DisplayProof\\}
  \caption{Typing rules for stacks and states of the machine}
  \label{fig:PCFstacks-typing-rules}
\end{figure}

\subsubsection{Denotational semantics of stacks and states}
If $\Tseqst\Gamma\pi\sigma$ (with
$\Gamma=(x_1:\sigma_1,\dots,x_k:\sigma_k)$) then, setting
$C=\Bwith_{i=1}^k\Extc{\Tsem{\sigma_i}}$, we define
$\Psemst\pi\Gamma\in\EPCOH(\Tens{\Excl C}{\Tsem\sigma},\Llbot)$ and if
$\Tseqstate\Gamma q$ then
$\Psemstate q\Gamma\in\EPCOH(\Excl C,\Llbot)$.  We know that
$\Psemst\pi\Gamma$ is fully characterized by the associated function
$\Fun{\Psemst\pi\Gamma}:\Pcoh
C\times\Pcoh{\Extc{\Tsem\sigma}}\to\Intercc 01$ (non-linear in the
first parameter and linear in the second one). And
$\Psemstate q\Gamma$ is fully characterized by the associated function
$\Fun{\Psemstate q\Gamma}:\Pcoh C\to\Intercc 01$.

\Itmath
If $\pi=\Stempty$, so that $\sigma=\Tnat$ then $\Psem\pi\Gamma$ is the
following composition of morphisms ($\Scatch\Nat$
is defined at the end of Section~\ref{sec:flate-general}):
\begin{tikzcd}
  \Tens{\Excl C}{\Flate\Nat}\ar[r,"\Weak C"]
  & \Flate\Nat\ar[r,"\Scatch\Nat"]
  & \Llbot
\end{tikzcd}.  The associated function
$f:\Pcoh C\times\Pcohp{\Extc{\Flate\Nat}}\to\Intercc 01$ is given by
$f(\Vect u,x)=x_\Nerr$.

\Itmath
If $\pi=\Stsucc\rho$ so that $\sigma=\Tnat$ and $\Tseqst\Gamma\rho\Nat$
then $\Psemst\pi\Gamma$ is defined as the following composition of
morphisms:
\begin{tikzcd}
  \Tens{\Excl C}{\Flate\Nat}\ar[r,"\Tens{\Excl C}{\Bfune\Fsucc}"]
  &\Tens{\Excl C}{\Flate\Nat}\ar[r,"\Psemst\rho\Gamma"]
  &\Llbot
\end{tikzcd}
and the semantics of $\Stpred\rho$ is defined similarly, using
$\Fpred$ instead of $\Fsucc$. The associated function
$f:\Pcoh C\times\Pcohp{\Extc{\Flate\Nat}}\to\Intercc 01$ is given by
$f(\Vect u,x)=g(\Vect u,\Matappa{\Bfune\Fsucc}x)$ where
$g=\Fun{\Psemst\rho\Gamma}$ (remember that
$\Matappa{\Bfune\Fsucc}x=x_\Nerr\Base\Nerr+\sum_{n\in\Nat}x_n\Base{n+1}$). The
case $\pi=\Stpred\rho$ is similar, replacing $\Fsucc$ with $\Fpred$.

\Itmath Assume that $\pi=\Stif {P_1}{P_2}\rho$ so that $\sigma=\Tnat$,
$\Tseqst\Gamma\rho\Tnat$ and $\Tseq\Gamma{P_i}\Tnat$ for $i=1,2$.  For
$n\in\Nat$, let $t_n\in\Kl\EPCOH(C,\Flate\Nat)$ be defined by
$t_0=\Psem{P_1}\Gamma$ and $t_{n+1}=\Psem{P_2}\Gamma$ for each
$n\in\Nat$. Let
$t=\Tuple{t_n}_{n\in\Nat}\in\Kl\EPCOH(C,(\Flate\Nat)^\Nat)$.  Then
$\Psemst\pi\Gamma$ is interpreted as the following composition of
morphisms:\\
\begin{tikzcd}
  \Tens{\Excl C}{\Flate\Nat}\ar[r,"\Contr C\ITens\Id"]
  &[2em] \Excl C\ITens\Excl C\ITens\Flate\Nat\ar[r,"\Excl C\ITens\Sym"]
  & \Excl C\ITens\Flate\Nat\ITens\Excl C\ar[r,"\Id\ITens t"]
  & \Excl C\ITens\Flate\Nat\ITens\Exclp{(\Flate\Nat)^\Nat}
  \ar[d,"\Id\ITens\Pcasee{\Nat,\Nat}"]\\
  && \Llbot
  & \Excl C\ITens\Flat\Nat\ar[l,"\Psem\rho\Gamma"]
\end{tikzcd}\\
Let $h_i=\Fun{\Psem{P_i}\Gamma}:\Pcoh C\to\Pcohp{\Extc{\Flate\Nat}}$
for $i=1,2$.  The function associated with $\Psemst\pi\Gamma$ is
$f:\Pcoh C\times\Pcohp{\Extc{\Flate\Nat}}\to\Intercc 01$ given by
$f(\Vect u,x)=g(\Vect u,x_\Nerr\Base\Nerr+x_0h_1(\Vect u)+
(\sum_{n=1}^\infty x_n)h_2(\Vect
u))$
% $=x_\Nerr g(\Vect u,\Base\Nerr)+x_0 g(\Vect u,h_1(\Vect
% u))+(\sum_{n=1}^\infty x_n)g(\Vect u,h_2(\Vect u))$ by linearity,
where $g=\Fun{\Psemst\rho\Gamma}$.

\Itmath Assume that $\pi=\Stlet xP\rho$ so that $\sigma=\Tnat$,
$\Tseqst\Gamma\rho\Tnat$ and $\Tseq{\Gamma,x:\Tnat}P\Tnat$. Then we
have
$\Psem P{\Gamma,x:\Tnat}\Compl\Seelyt:\Excl
C\ITens\Excl{\Flate\Nat}\to\Flate\Nat$ and hence
$t=\Slete{\Psem P{\Gamma,x:\Tnat}\Compl\Seelyt}\in\EPCOH(\Excl
C\ITens\Flate\Nat,\Flate\Nat)$. Then $\Psemst\pi\Gamma$ is defined as
the following composition of morphisms:\\
\begin{tikzcd}
  \Excl C\ITens\Flate\Nat\ar[r,"\Contr C\ITens\Id"]
  &[2em] \Excl C\ITens\Excl C\ITens\Flate\Nat\ar[r,"\Id\ITens t"]
  & \Excl C\ITens\Flate\Nat\ar[r,"\Psemst\rho\Gamma"]
  & \Llbot
\end{tikzcd}.\\
Let
$h=\Fun{\Psem P{\Gamma,x:\Tnat}}:\Pcoh
C\times\Pcohp{\Extc{\Flate\Nat}}\to\Pcohp{\Extc{\Flate\Nat}}$.  The
function associated with $\Psemst\pi\Gamma$ is
$f:\Pcoh C\times\Pcohp{\Extc{\Flate\Nat}}\to\Intercc 01$ given by
$f(\Vect u,x)=g(\Vect u,x_\Nerr\Base\Nerr+\sum_{n=0}^\infty x_nh(\Vect
u,\Base n))$
% $=x_\Nerr g(\Vect u,\Base\Nerr)+\sum_{n=0}^\infty x_ng(\Vect
% u,h(\Vect u,\Base n))$ by linearity and continuity,
where $g=\Fun{\Psemst\rho\Gamma}$.

\Itmath We end with the semantics of states so assume that
$q=\State M\pi$ with $\Tseq\Gamma M\sigma$ and
$\Tseqst\Gamma\pi\sigma$ so that $\Tseqstate\Gamma q$. Then we have
$\Psem M\Gamma:\Excl C\to\Tsem\sigma$ and
$\Psemst\pi:\Excl C\ITens\Tsem\sigma\to\Llbot$, and
$\Psemstate q\Gamma:\Excl C\to\Llbot$ is the following composition of
morphisms
\begin{tikzcd}
  \Excl C\ar[r,"\Contr C"]
  &\Excl C\ITens\Excl C\ar[r,"\Id\ITens\Psem M\Gamma"]
  &[1.4em]\Excl C\ITens\Tsem\sigma\ar[r,"\Psemst\pi\Gamma"]
  &\Llbot
\end{tikzcd}. The function $f:\Pcoh C\to\Intercc 01$ associated with
$\Psemstate q\Gamma$ is given by $f(\Vect u)=g(\Vect u,h(\Vect u))$
where $g=\Fun{\Psemst\pi\Gamma}$ and $h=\Fun{\Psem M\Gamma}$.

\begin{lemma}\label{lemma:PCF-stack-interp-linearity}
  If $\Tseq\Gamma\pi\Tnat$ then the associated function
  $g=\Fun{\Psemst\pi\Gamma}:\Pcoh
  C\times\Pcohp{\Extc{\Flate\Nat}}\to\Intercc 01$ is Scott-continuous,
  and linear in its last argument $\in\Pcohp{\Extc{\Flate\Nat}}$ and
  satisfies $g(\Vect u,\Base\Nerr)=1$.
\end{lemma}
\begin{proof}
  Linearity and continuity results from the fact that
  $\Psemst\pi\Gamma\in\EPCOH(\Excl{\Tsem\Gamma}\ITens\Flate\Nat,\Llbot)$. The
  property $g(\Vect u,\Base\Nerr)=1$ results from a simple inspection
  of the description above of the semantics of stacks.
\end{proof}

\subsubsection{The machine.}
Remember that a term $M$ is \emph{almost closed} if
$\Tseq\Gamma M\sigma$ where $\Gamma$ is a ground context (all types
appearing in $\Gamma$ are $\Tnat$); similarly a stack $\pi$ is almost
closed if $\Tseqst\Gamma\pi\sigma$ where $\Gamma$ is ground. Notice
that a ground context can be identified with a finite set $V$ of
variables (the variables which are assigned the type $\Tnat$ by
$\Gamma$).

We define a function $\Kreval$ (a priori partial) from almost closed
typed states $\Polynomts\Realpc V\Natt$. The definition is given in
Figure~\ref{fig:PCF-Krivine-machine}.  Notice that there is no rule
for terms of shape $\Fix M$ and hence our machine is restricted to
\Fixfree{} states; the corresponding rule should be
$\Kreval\State{\Fix M}{\pi}=\Kreval\State{M}{\Starg{\Fix M}\pi}$. The
main reason is that, with this additional rule,
Lemma~\ref{lemma:Kreval-total} does not hold anymore.

\begin{figure}{\footnotesize
  \begin{alignat*}3
    &\Kreval\State{\Errconv}{\pi}=\Ptone
    &\quad
    &\Kreval\State{\Errdiv}\pi=\Ptzero
    \\
    &\Kreval\State{\Dice r}\pi=
    \Ptcl r{\Kreval\State{\Num 0}\pi}{(1-r)}{\Kreval\State{\Num 1}\pi}
    &\quad
    &\Kreval\State x\pi=\Ptvar x\cS
    \text{ where }\cS(\Nerr)=\Ptone \text{ and }
    \cS(n)=\Kreval\State{\Num n}\pi
%   &\Kreval\State x\pi=x(\Nerr)+\sum_{n\in\Nat}x(n)\Kreval\State{\Num n}\pi
    \\
    &\Kreval\State{\Num n}\Stempty=0
    &\quad
    &\Kreval\State{\Num n}{\Stsucc\pi}=\Kreval\State{\Num{n+1}}\pi
    \\
    &\Kreval\State{\Num 0}{\Stpred\pi}=\Kreval\State{\Num 0}{\pi}
    &\quad
    &\Kreval\State{\Num{n+1}}{\Stpred\pi}=\Kreval\State{\Num n}{\pi}
    \\
    &\Kreval\State{\Num 0}{\Stif PQ\pi}=\Kreval\State P\pi
    &\quad
    &\Kreval\State{\Num{n+1}}{\Stif PQ\pi}=\Kreval\State Q\pi\\
    &\Kreval\State{\Let xMN}\pi=\Kreval\State M{\Stlet xN\pi}
    &\quad
    &\Kreval\State{\Num n}{\Stlet xN\pi}=\Kreval\State{\Subst N{\Num n}x}\pi
    \\
    &\Kreval\State{\App MN}\pi=\Kreval\State M{\Starg N\pi}
    &\quad
    &\Kreval\State{\Abst x\sigma M}{\Starg N\pi}=\Kreval\State{\Subst MNx}{\pi}
    \\
    &\Kreval\State{\Fix M}{\pi}=\Kreval\State M{\Starg{\Fix M}\pi}
    &\quad
    &%\Kreval\State{\Dice r}\pi=r\Kreval\State{\Num 0}\pi
    % +(1-r)\Kreval\State{\Num 1}\pi
    % \Kreval\State{\Fix M}{\pi}=\Kreval\State{M}{\Starg{\Fix M}\pi}
    &&
  \end{alignat*}}
  \vspace{-1.8em}
  \caption{A Krivine function from almost closed states to polynomials}
  \label{fig:PCF-Krivine-machine}
\end{figure}

\subsection{Main property of the machine}

We prove now that this function $\Kreval$ is total on almost closed
\Fixfree{} states and yields elements of $\Polynomts\Realp
V\Natt$. The proof is by reducibility, following the general format
used by Jean-Louis Krivine in his work on classical realizability.

Let $V$ be a finite set of variables considered as a ground
context\footnote{This requires a total ordering on the elements of
  $V$, we keep this further information implicit identifying $V$ with
  such a sequence.}, we define first a \emph{pole} $\Pole V$ which is
the set of all \Fixfree{} states $q$ such that
\begin{itemize}
\item $\Tseqstate Vq$
\item $\Kreval(q)$ is a well-defined tree $S$ which belongs to
  $\Polynomts\Realp V\Natt$
\item and the associated polynomial
  $\Polyofpt(S)$ belongs to $\EPCOH((\Flate\Nat)^V,\Llbot)$ and satisfies
  $\Polyofpt(S)=\Psemstate qV$.
\end{itemize}

With each type $\sigma$ we associate a set $\Redintst\sigma$ of stacks
$\pi$ such that $\Tseqst V\pi\sigma$ by:
\begin{itemize}
\item $\Redintst\Tnat$ is the set of all \Fixfree{} stacks $\pi$ such that
  $\Tseqst{V}\pi\Tnat$ and
  $\forall n\in\Nat\ \State{\Num n}\pi\in\Pole V$. Notice that in
  particular $\Stempty\in\Redintst\Tnat$.
\item and
  $\Redintst{\Timpl\sigma\tau}=\Eset{\Starg M\pi\St M\in\Redint\sigma
    \text{ and } \pi\in\Redintst\tau}$ where
  $\Redint\sigma=\Eset{M\in\Termstyf{V}\sigma\St\forall\pi\in\Redintst\sigma\
    \State M\pi\in\Pole V}$. Here $\Termstyf{V}\sigma$ is the
  set of all \Fixfree{} terms $M$ such that $\Tseq VM\sigma$.
\end{itemize}

\begin{lemma}\label{lemma:push-int-stack}
  If $\pi\in\Redintst\Tnat$ then
  $\Stsucc\pi,\Stpred\pi\in\Redintst\Tnat$. If $P,Q\in\Redint\Tnat$
  then $\Stif PQ\pi\in\Redintst\Tnat$. If $\Tseq{V,y:\Tnat}P\Tnat$
  and $\forall n\in\Nat\ \Subst P{\Num n}y\in\Redint\Tnat$, then
  $\Stlet yP\pi\in\Redintst\Tnat$.
\end{lemma}

\begin{lemma}\label{lemma:redint-ground-type}
  We have $V\subseteq\Redint\Tnat$ and
  $\forall n\in\Nat\ \Num n\in\Redint\Tnat$. Last
  $\Errdiv,\Errconv\in\Redint\Tnat$ and for all
  $r\in\Intercc 01\cap\Rational$ one has $\Dice r\in\Redint\Tnat$.
\end{lemma}
\begin{proof}
  Let $y\in V$ and $\pi\in\Redintst\Tnat$, we have to prove that
  $q=\State y\pi\in\Pole V$. We have $\Kreval(q)=\Ptvar y\cS$ where
  $\cS(n)=\Kreval\State{\Num n}\pi$ for each $n\in\Nat$ and
  $\cS(\Nerr)=\One$. By definition of $\Pole V$ we have
  $\forall n\in\Nat\ \cS(n)\in\Polynomts\Realp V\Natt$ and hence
  $\Kreval(q)$ is well-defined and belongs to
  $\Polynomts\Realp V\Natt$. Setting $C=\Extc{\Flate\Nat}^V$, the map
  $f=\Fun{\Psemstate qV}:\Pcoh C\to\Intercc 01$ is characterized by
  $f(\Vect u)=g(\Vect u,u(y))$ (remember that $\Vect u$ is a vector
  $(u(z))_{z\in V}$ of elements of $\Pcohp{\Extc{\Flate\Nat}}$) where
  $g=\Fun{\Psemst\pi\Gamma}$. On the other hand the polynomial
  associated with $\Kreval(q)$ is
  $\Polyofpt(\Kreval(q))=y_\Nerr+\sum_{n\in\Nat}y_n\Polyofpt(\cS(n))$
  which defines the function $f':\Pcoh C\to\Realpc$ by
  $f'(\Vect u)=u(y)_\Nerr+\sum_{n\in\Nat}u(y)_nf_n(\Vect u)$ where we
  know that
  $f_n(\Vect u)=\Fun{\Psemstate{\State{\Num n}{\pi}}V}(\Vect u)$ by
  our assumption that $\pi\in\Redintst\Tnat$, that is
  $f_n(\Vect u)=g(\Vect u,\Base n)$ by definition of the
  interpretation of states. This proves
  $f'(\Vect u)=g(\Vect u,u(y))=f(\Vect u)$ by
  Lemma~\ref{lemma:PCF-stack-interp-linearity}.

  Let $n\in\Nat$ and $\pi\in\Redintst\Tnat$, by definition of
  $\Redint\Tnat$ we have $\State{\Num n}{\pi}\in\Pole V$ and hence
  $\Num n\in\Redint\Tnat$.

  Let $\pi\in\Redintst\Tnat$, then $\Kreval\State\Errdiv\pi=\Ptzero$
  and $\Kreval\State\Errconv\pi=\Ptone$ are well-defined and belong to
  $\Polynomts\Realp V\Natt$. The map
  $g=\Fun{\Psemst\pi V}:\Pcoh
  C\times\Pcohp{\Extc{\Flate\Nat}}\to\Intercc 01$ satisfies
  $g(\Vect u,0)=0$ by Lemma~\ref{lemma:PCF-stack-interp-linearity} and
  since $\Psem\Errdiv V(\Vect u)=0$ we have
  $\Psemstate{\State\Errdiv\pi}V(\Vect u)=0$ and hence
  $\Psemstate{\State\Errdiv\pi}V=\Polyofpt(\Kreval\State\Errdiv\pi)$. We
  deal similarly with $\Errconv$ since
  $\Psem\Errconv V(\Vect u)=\Base\Nerr$ and $g(\Vect u,\Base\Nerr)=1$
  by Lemma~\ref{lemma:PCF-stack-interp-linearity}. Last we have
  $\Kreval\State{\Dice r}\pi=\Ptcl{r}{\Kreval\State{\Num
      0}\pi}{(1-r)}{\Kreval\State{\Num 1}\pi}$ and by our assumption
  about $\pi$ we have $\State{\Num 0}\pi,\State{\Num 1}\pi\in\Pole
  V$. It follows that $S=\Kreval\State{\Dice r}\pi$ is well-defined
  and belongs to $\Polynomts\Realp V\Natt$. Moreover
  $\Polyofpt(S)=r\Polyofpt(S_0)+(1-r)\Polyofpt(S_1)$ where
  $S_n=\Kreval\State{\Num n}\pi$ for $n=0,1$. We also know that
  $\Polyofpt(S_n)=\Psemstate{\State{\Num n}\pi}V$ for $n=0,1$. With
  the same notations as above we have
  $\Psemstate{\State{\Dice r}\pi}V(\Vect u)=g(\Vect u,r\Base
  0+(1-r)\Base 1)=rg(\Vect u,\Base 0)+(1-r)g(\Vect u,\Base 1)$ by
  Lemma~\ref{lemma:PCF-stack-interp-linearity}. This ends the proof
  that $\Polyofpt(S)=\Psemstate{\State{\Dice r}\pi}V$ since
  $\Fun{\Psemstate{\State{\Num n}\pi}V}(\Vect u)=g(\Vect u,\Base n)$.
\end{proof}

\begin{lemma}\label{lemma:Kreval-total}
  If $M$ is \Fixfree{} and if
  $\Tseq{V,x_1:\sigma_1,\dots,x_k:\sigma_k}M\tau$ and if
  $N_j\in\Redint{\sigma_j}$ for $j=1,\dots,k$, then
  $\Substbis M{\Vect N/\Vect x}\in\Redint\tau$.
\end{lemma}
\begin{proof}
  By induction on the proof of the typing judgment
  $\Tseq{V,\Gamma}M\tau$ where
  $\Gamma=(x_1:\sigma_1,\dots,x_k:\sigma_k)$.  We use $M'$ for
  $\Substbis M{\Vect N/\Vect x}$ to increase readability.

  In cases $M=\Num n$, $M=\Dice r$, $M=\Errdiv$, $M=\Errconv$ and
  $M=y\in V$ we have $M'=M$ and $M\in\Redint\Tnat$ by
  Lemma~\ref{lemma:redint-ground-type}.

  Assume that $M=x_i$ for some
  $i\in\Eset{1,\dots,k}$, we have $\tau=\sigma_i$ and $M'=N_i$ so that
  $M'\in\Redint\tau$ by our assumption about $N_i$.

  Assume that $M=\If RPQ$, so that $\tau=\Tnat$, $M'=\If{R'}{P'}{Q'}$
  and, by inductive hypothesis, $R',P',Q'\in\Redint\Tnat$. Let
  $\pi\in\Redintst\Tnat$, we have
  $\Kreval\State{M'}\pi=\Kreval\State{R'}{\Stif{P'}{Q'}\pi}$ which is
  well-defined and belongs to $\Polynomts\Realp V\Natt$ by inductive
  hypothesis since $\Stif{P'}{Q'}\pi\in\Redintst\Tnat$ by
  Lemma~\ref{lemma:push-int-stack}. We end the proof that
  $\State{M'}\pi\in\Pole V$ by observing that
  $\Psemstate{\State{M'}\pi}V=\Psemstate{\State{R'}{\Stif{P'}{Q'}\pi}}V$.

  The cases $M=\Succ N$ and
  $M=\Pred N$ are similar and simpler.

  Assume that $M=\Let xRP$ so that $\tau=\Tnat$ and
  $M'=\Let x{R'}{P'}$. Let $\pi\in\Redintst\Tnat$, we have
  $\Kreval\State{M'}{\pi}=\Kreval\State{R'}{\Stlet x{P'}\pi}$. By
  inductive hypothesis we have $R'\in\Redint\Tnat$ and
  $\Subst{P'}{\Num n}x\in\Redint\Tnat$ since $\Num n\in\Redint\Tnat$
  for all $n\in\Nat$. Therefore $\Stlet x{P'}\pi\in\Redintst\Tnat$ by
  Lemma~\ref{lemma:push-int-stack}. It follows that
  $\Kreval\State{M'}{\pi}$ is well-defined and belongs to
  $\Polynomts\Realp V\Natt$. We end the proof that
  $\State{M'}\pi\in\Pole V$ by observing that
  $\Psemstate{\State{M'}\pi}V=\Psemstate{\State{R'}{\Stlet x{P'}\pi}}V$.

  Assume that $M=\App RP$ with $\Tseq{V,\Gamma}R{\Timpl\sigma\tau}$
  and $\Tseq{V,\Gamma}P{\sigma}$, we have $M'=\App{R'}{P'}$ and, by
  inductive hypothesis, $R'\in\Redint{\Timpl\sigma\tau}$ and
  $P'\in\Redint\sigma$. Let $\pi\in\Redintst\tau$, we have
  $\Starg{P'}\pi\in\Redintst{\Timpl\sigma\tau}$ by definition of this
  latter set and hence
  $\Kreval\State{M'}\pi=\Kreval\State{R'}{\Starg{P'}\pi}$ is
  well-defined and belongs to $\Polynomts{\Realp}{V}{\Natt}$.  We end
  the proof that $\State{M'}\pi\in\Pole V$ by observing that
  $\Psemstate{\State{M'}\pi}V=\Psemstate{\State{R'}{\Starg{P'}\pi}}V$.

  Assume last that $M=\Abst x\sigma P$ with
  $\Tseq{V,\Gamma,x:\sigma}{P}{\phi}$ and
  $\tau=(\Timpl\sigma\phi)$. Let $\rho\in\Redintst{\Timpl\sigma\phi}$,
  that is $\rho=\Starg N\pi$ with $N\in\Redint\sigma$ and
  $\pi\in\Redintst\phi$. By inductive hypothesis applied to $P$, we
  have that $\Subst{P'}Nx\in\Redint\phi$ and hence
  $\Kreval\State{M'}{\rho}
  =\Kreval\State{\Subst{P'}Nx}{\pi}\in\Polynomts{\Realp}{V}{\Natt}$. We
  end the proof that $\State{M'}\pi\in\Pole V$ by observing that
  $\Psemstate{\State{M'}{\Starg
      N\pi}}V=\Psemstate{\State{\Subst{P'}Nx}\pi}V$.
\end{proof}

\begin{theorem}
  Assume that $\Tseq VM\Tnat$ and that $M$ is \Fixfree. Then
  $\Kreval\State M\Stempty$ is a well-defined element $S$ of
  $\Polynomts\Realp V\Natt$ which satisfies
  $\Polyofpt(S)=\Scatch\Compl\Psem MV$, that is, for all
  $\Vect u\in\Pcoh{\Extc{\Tsem V}}$ one has
  $\Fun{\Polyofpt(S)}(\Vect u)=\Fun{\Psem MV}(\Vect u)_\Nerr$.
\end{theorem}
This is the special case of the above lemma when $\Gamma$ is the empty
context.

\subsection{Application}
Let $M$ be such that $\Tseq VM\Tnat$ (which typically can contain
fixpoint constructs). Then if $\Tseq V{M_0,M^0}\Tnat$ are \Fixfree{}
and satisfy $M_0\Termso M\Termso M^0$ then
$S_0=\Kreval\State{M_0}{\Stempty}$ and
$S^0=\Kreval\State{M^0}{\Stempty}$ are elements of
$\Polynomts\Realp V\Natt$ which satisfy
$\Polyofpt(S^0)=\Scatch\Compl\Psem{M^0}V$ and
$\Polyofpt(S_0)=\Scatch\Compl\Psem{M_0}V$ and hence
$\Polyofpt(S_0) \Exto{\Limpl{\Excl{(\Flate\Nat})^V}\Llbot}
\Scatch\Compl\Psem{M}V
\Exto{\Limpl{\Excl{(\Flate\Nat})^V}\Llbot}\Polyofpt(S^0)$ by
Theorem~\ref{th:psem-ext-order}.

Of course the polynomials $\Polyofpt(S_0)$ and $\Polyofpt(S^0)$ are
usually infinite but for any finite subset $J$ of $\Nat$ we can
precompose the former with $(\Idl J\Nat)^V$ and the latter with
$(\Idu J\Nat)^V$ in $\Kl\EPCOH$ and one obtains by
Lemma~\ref{lemma:upper-id} in that way two finite polynomials $t_0$
and $t^0$ such that
$t_0 \Exto{\Limpl{\Excl{(\Flate\Nat})^V}\Llbot} \Scatch\Compl\Psem{M}V
\Exto{\Limpl{\Excl{(\Flate\Nat})^V}\Llbot}t^0$.  Notice that these
restrictions by $J$ can be performed on the fly during the run of
$\Kreval$ which will then return finite polynomials.

\begin{remark}
  This reducibility proof would also work for terms containing some
  restricted form of recursion such as the higher order primitive
  recursion of Gödel System T. It turns out that, for such terms, the
  support of the interpretation of type $\sigma$ in $\PCOH$ is a
  finitary set in the sense of the Finiteness Space
  semantics~\cite{Ehrhard00b}. In that case we can apply the method
  above to the term $M$ itself, without taking before syntactic
  approximations $M_0$ and $M^0$. We just need to precompose
  $(\Idl J\Nat)^V$ and $(\Idu J\Nat)^V$ for getting finite polynomial
  approximations.
\end{remark}

\section*{Related work and conclusion}
This work takes place in a general trend trying to extract formal
tools from denotational models, much in the spirit of Abstract
Interpretation. Typical developments of this kind are the various
intersection typing systems dating back to the early work of Coppo and
Dezani~\cite{CoppoDezzani78} which are often deeply related with
denotational models such as Scott semantics of the relational model of
LL and of the $\lambda$-calculus. Among these contributions one of the
most relevant to the present work is~\cite{BreuvartLago18} where an
intersection typing system is designed for approximating probabilities
of convergence. It is still an open problem to understand the
connection between this type-based approximations and those, based on
PCS, that we develop here. Another interesting connection might be
found in the work of Salvati and Walukiewicz on Higher Order Recursion
Schemes where Krivine machines play a essential role, in connection
with denotational properties of $\lambda$-terms with fixpoints.

More examples and practical computations will be the object of a further paper.

% Citer:
%
% Salvati et Walukiewicz pour leur MK qui calcule des points de
% la sémantique relationnelle

%%% Local Variables:
%%% mode: latex
%%% TeX-master: "preprint"
%%% End:

%\bibliographystyle{alpha}
\bibliography{newbiblio}

% \section*{Appendix: additional comments on Theorem~\ref{th:Tdistobs-denot-dist}}
% We would like to add some further observations on
% Theorem~\ref{th:Tdistobs-denot-dist} for the referees. If the paper is
% accepted these comments will be made available online.

% \input{comment.tex}

\end{document}